\documentclass[12pt,reqno]{amsart}

\usepackage{amsmath,amssymb,amsthm,mathtools}
\usepackage{centernot}
\usepackage[margin=1in]{geometry}
\usepackage{enumitem}
\usepackage{needspace}
\usepackage{xcolor}
\usepackage{url}
\usepackage{hyperref}
\usepackage[protrusion=false]{microtype}

\allowdisplaybreaks[3]
\numberwithin{equation}{section}
\setlist{leftmargin=2.2em,itemsep=0.25em,topsep=0.4em}

\definecolor{linkblue}{RGB}{25,77,135}
\hypersetup{
  colorlinks=true,
  linkcolor=linkblue,
  citecolor=linkblue,
  urlcolor=linkblue,
  unicode=true,
  pdfencoding=auto,
  pdftitle={Relative Entropy Decay via BKM coercivity for Quantum Markov Semigroups},
  pdfauthor={Li Gao and Jingyu Guo},
  pdfsubject={BKM coercivity, GNS symmetrization, and complete relative entropy decay
  without detailed balance},
  pdfkeywords={quantum Markov semigroup,Bogoliubov-Kubo-Mori metric,complete modified
  logarithmic Sobolev inequality,GNS symmetrization,Pimsner-Popa index,KMS
  symmetry,Gibbs sampler}
}

\newtheorem{theorem}{Theorem}[section]
\newtheorem{proposition}[theorem]{Proposition}
\newtheorem{lemma}[theorem]{Lemma}
\newtheorem{corollary}[theorem]{Corollary}

\theoremstyle{definition}
\newtheorem{definition}[theorem]{Definition}
\newtheorem{assumption}[theorem]{Assumption}
\newtheorem{example}[theorem]{Example}

\theoremstyle{remark}
\newtheorem{remark}[theorem]{Remark}

\newcommand{\Tr}{\operatorname{Tr}}
\newcommand{\id}{\operatorname{id}}
\newcommand{\Fix}{\operatorname{Fix}}
\newcommand{\Ran}{\operatorname{Ran}}
\newcommand{\Ker}{\operatorname{Ker}}
\newcommand{\spec}{\operatorname{spec}}
\newcommand{\EP}{\operatorname{EP}}
\newcommand{\B}{\mathrm B}
\newcommand{\cb}{\mathrm{cb}}

\newcommand{\cB}{\mathcal B}
\newcommand{\cD}{\mathcal D}
\newcommand{\cH}{\mathcal H}
\newcommand{\cK}{\mathcal K}
\newcommand{\cL}{\mathcal L}
\newcommand{\cN}{\mathcal N}
\newcommand{\cR}{\mathcal R}

\newcommand{\eps}{\varepsilon}
\newcommand{\dd}{\,\mathrm d}
\newcommand{\ip}[2]{\left\langle #1,#2\right\rangle}
\newcommand{\norm}[1]{\left\lVert #1\right\rVert}
\providecommand{\diag}{\operatorname{diag}}

\title[Relative Entropy Decay via BKM coercivity]{Relative Entropy Decay via BKM
coercivity for Quantum Markov Semigroups}
\author{Li Gao}
\address{School of Mathematics and Statistics, Wuhan University, Wuhan 430072, China}
\address{Wuhan Institute of Quantum Technology, Wuhan 430075, China}
\email{gao.li@whu.edu.cn}
\author{Jingyu Guo}
\address{School of Mathematics and Statistics, Wuhan University, Wuhan 430072, China}
\email{guojingyu@whu.edu.cn}
\date{September 14, 2026}

\begin{document}
\begin{abstract}
We introduce the notion of the BKM coercivity constant and show that the positivity of
this constant governs the exponential relative entropy decay of the semigroup. Indeed,
we prove that the modified log-Sobolev constant is equivalent to the BKM coercivity
constant up to a constant depending on the asymptotic conditional expectation.
We also discover that the BKM coercivity constants for matrix amplifications are
already attained with a qubit auxiliary system, which also coincides with the GNS
spectral gap of the generating Lindbladian. As a corollary, we establish a criterion
that the complete modified
log-Sobolev inequality holds if and only if the GNS gap of the semigroup is strictly
positive. All the discussion above applies to quantum Markov semigroups that admit a
faithful asymptotic conditional expectation as long-time equilibration, with no detailed
balance condition assumed. As an application, we show the finite-dimensional
Chen-Kastoryano-Gily\'en Gibbs samplers always satisfy the complete modified log-Sobolev
inequality.
\end{abstract}

\maketitle

\section{Introduction}
\label{sec:introduction}

Quantum Markov semigroups (QMS) describe the dissipative evolution of time-homogeneous
open quantum systems. In recent decades, functional inequalities for quantum Markov
semigroups have been intensively studied as tools for quantifying convergence to
equilibrium \cite{OlkiewiczZegarlinski1999,KastoryanoTemme2013,
CarlenMaas2017,Bardet2017Decoherence,
GaoJungeLaRacuente2020fisher,GaoRouze2022}. For a semigroup $P_t=e^{t\cL}$
with asymptotic conditional expectation $E$, the modified logarithmic Sobolev inequality
(MLSI) characterizes the exponential decay of relative entropy
\begin{align}
  D(P_{t*}\rho\Vert E_*\rho)
  \le e^{-2\alpha t}D(\rho\Vert E_*\rho),
  \qquad t\ge0,
  \label{eq:intro-entropy-decay}
\end{align}
for some $\alpha>0$ and every state $\rho$ \cite{Spohn1978,Bardet2017Decoherence}. Here
$D$ denotes quantum relative entropy. The maps $P_{t*}$ and $E_*$ are the trace
adjoints of $P_t$ and $E$, acting on states in the Schr\"odinger picture.
Motivated by tensorization and applications to many-body systems, Gao, Junge, and
LaRacuente introduced the complete modified logarithmic Sobolev inequality (CMLSI)
\cite{GaoJungeLaRacuente2020fisher}. It requires exponential decay
\eqref{eq:intro-entropy-decay} for every amplification $\id_{M_n}\otimes P_t$, with a
decay rate independent of $n$. This condition controls entropy decay even when the
system is initially entangled with an ancillary system.

In finite dimensions, Gao and Rouz\'e proved positivity of the CMLSI constant under
Gelfand--Naimark--Segal (GNS) symmetry, a quantum analog of the detailed balance
condition
\cite{GaoRouze2022}. Their estimate gives a quantitative lower bound in terms of the
spectral gap and the complete Pimsner--Popa index. Later, using complete positivity
order, Gao, Junge, LaRacuente, and Li obtained bounds with optimal logarithmic
dependence on this index \cite{GaoJungeLaRacuenteLi2025}. These comparisons provide a
route from spectral gap to relative entropy decay under GNS symmetry.

Davies semigroups are a natural setting for this approach. They arise as thermalization
models in the weak-coupling limit \cite{Davies1974}, and their dissipative generators
satisfy GNS symmetry. For one-dimensional commuting Hamiltonians, MLSI has been used to
prove rapid thermalization \cite{BardetCapelGaoLuciaPerezGarciaRouze2023}. Subsequent
work established MLSI lower bounds that are uniform in system size
\cite{KochanowskiAlhambraCapelRouze2025}. Related results cover CSS models under a
Dobrushin--Shlosman condition \cite{StengeleEtAl2026CSS} and two-dimensional Abelian
quantum double models \cite{StengeleEtAl2026QuantumDouble}. These applications
illustrate the effectiveness of entropy methods for commuting many-body systems.

For noncommuting Hamiltonians, a major breakthrough was the construction of an
efficiently simulable Gibbs sampler by Chen, Kastoryano, and Gily\'en \cite{CKG2025}.
Chen, Kastoryano, Brand\~ao, and Gily\'en developed this approach as a framework for
quantum thermal simulation \cite{CKBG2025}. The resulting Lindblad dynamics preserve the
exact Gibbs state through Kubo--Martin--Schwinger (KMS) symmetry, a weaker detailed
balance condition than the GNS one.

The CKG construction replaces exact resolution of Bohr frequencies with
Gaussian-filtered Heisenberg evolution. For local Hamiltonians and local input
operators, Lieb--Robinson bounds ensure that the filtered jump operators are
quasi-local, which is a crucial property that enables the efficient simulation on a
quantum computer. The resulting generator is generally not GNS symmetric. This
construction demonstrates the practical value of the weaker KMS condition: exact Gibbs
stationarity can be combined with quasi-local dynamics and efficient quantum simulation
\cite{CKG2025,CKBG2025}.

Given the efficient simulation, it still leaves the mixing time of the dynamics to be
estimated. At high temperatures, Rouz\'e, Fran\c{c}a, and Alhambra first established
polynomial-time Gibbs state preparation \cite{RouzeFrancaAlhambra1}, and later they
obtained optimal rapid mixing through exponential decay of an oscillator norm
\cite{RouzeFrancaAlhambra2}. For one-dimensional finite-range Hamiltonians, Bergamaschi
and Chen proved a positive lower bound on the KMS spectral gap at every fixed positive
temperature \cite{BergamaschiChen2026}. The bound is uniform in system size and includes
the Heisenberg chain. These results motivate the search for comparable relative entropy
decay estimates beyond GNS symmetry.

Entropy methods in this setting face difficulties that are absent under GNS symmetry.
Hamiltonian evolution can prevent CMLSI even in the presence of dissipation
\cite[Proposition~1.1]{LaRacuente2026SelfRestricting}. For finite-dimensional semigroups
with unital preduals, LaRacuente nevertheless proved complete entropy contraction toward
the decoherence-free subspace at fixed positive time scales
\cite[Theorem~1.2]{LaRacuente2026SelfRestricting}. These results do not provide a
general MLSI or CMLSI criterion under KMS symmetry alone.

For the CKG sampler, efficient implementation and KMS spectral gap estimates do not by
themselves establish relative entropy decay of the form \eqref{eq:intro-entropy-decay}.
A central question is therefore whether MLSI and CMLSI can be established for these
exact Gibbs dynamics beyond GNS symmetry. In this work, we obtain quantitative MLSI and
CMLSI bounds for finite-dimensional QMS admitting a faithful asymptotic conditional
expectation. Our criteria for positivity require no detailed balance. As an application,
we prove that every fixed finite-dimensional CKG sampler has a positive CMLSI constant
relative to its fixed-point algebra.

\subsection{Main results}
Throughout this subsection, let $P_t=e^{t\cL}$ be a QMS on $\B(\cH)$ with
$\dim\cH<\infty$. Assume that $\displaystyle \lim_{t\to \infty}P_t=E$ converges in norm
to a faithful conditional expectation $E$. Write $\cN=\Ran E=\Fix(P_t)$. We assume no
detailed balance condition.

At a faithful invariant state $\omega$, the Hessian of relative entropy defines the
Bogoliubov--Kubo--Mori (BKM) metric
\[
  \langle A,A \rangle_{\mathrm{BKM},\omega }=\Tr[A\Gamma_\omega^{-1}(A)]
\]
where $\Gamma_\omega^{-1}$ is the inverse of
\[
  \Gamma_\omega(X)=\int_0^1\omega^sX\omega^{1-s}\dd s.
\]
The BKM coercivity constant is
\[
  \lambda_{\mathrm{BKM}}(\cL)
  :=\inf_{\omega,A}
  \frac{-\Tr[\cL_*(A)\Gamma_\omega^{-1}(A)]}
  {\Tr[A\Gamma_\omega^{-1}(A)]}.
\]
The infimum ranges over faithful invariant states $\omega$ and nonzero self-adjoint $A$
satisfying $E_*A=0$ (see Definition~\ref{def:bkm-constants}). Let $\alpha(\cL)$ denote
the optimal constant in \eqref{eq:intro-entropy-decay}. Recall that the Pimsner--Popa
index is given by the optimal order constant
\[
  C(E)=\inf \{ C>0 \ |\  \rho\le CE_*\rho \text{ for all states $\rho$ }\} \ .
\]

\begin{theorem}[MLSI and BKM coercivity]
\label{thm:intro-mlsi}
Under the above assumptions,
\begin{equation}
  \frac{\lambda_{\mathrm{BKM}}(\cL)}{2C(E)}
  \le \alpha(\cL)
  \le \lambda_{\mathrm{BKM}}(\cL).
  \label{eq:intro-mlsi-bounds}
\end{equation}
\end{theorem}

Thus MLSI holds exactly when $\lambda_{\mathrm{BKM}}(\cL)>0$. The lower bound uses
Frenkel's integral formulas for relative entropy and its Hessian \cite[Theorems~6
and~7(a)]{Frenkel2023}. Under GNS symmetry, $\lambda_{\mathrm{BKM}}(\cL)$ equals the
usual spectral gap of $-\cL$. The theorem therefore recovers the ordinary index bound of
\cite[Theorem~3.3]{GaoRouze2022} in that case.

We next adjoin an ancillary system. For $n\ge1$, write
\[
  \cL^{(n)}:=\id_{M_n}\otimes\cL,
  \qquad E^{(n)}:=\id_{M_n}\otimes E.
\]
We define the complete BKM coercivity constant as
\begin{align*}
  \lambda_{\mathrm{BKM}}^{\mathrm c}(\cL)
  &:=\inf_{n\ge1}\lambda_{\mathrm{BKM}}(\cL^{(n)}).
\end{align*}
We show that the complete BKM coercivity constant coincides with the GNS spectral gap.
For a general QMS, this gap is defined without a symmetry assumption by
\[
  \lambda_{\mathrm{GNS}}(\cL)
  :=\inf_{\substack{X\ne0\\E(X)=0}}
  \frac{-\Re\Tr(\omega X^*\cL(X))}{\Tr(\omega X^*X)}.
\]

\begin{theorem}[BKM coercivity under amplification]
\label{thm:intro-amplification}
Under the above assumptions,
\begin{equation}
  \lambda_{\mathrm{BKM}}^c(\cL)=
  \lambda_{\mathrm{BKM}}(\id_{M_2}\otimes\cL)
  =\lambda_{\mathrm{GNS}}(\cL)
  \le\lambda_{\mathrm{BKM}}(\cL).
  \label{eq:intro-bkm-amplification}
\end{equation}
\end{theorem}

Without detailed balance, Wirth compared the GNS gap with gaps associated with
operator-monotone functions \cite[Corollary~3.7 and Remark~3.9]{WirthKMSGNS2026}. His
comparison includes the BKM metric and gives the lower bound by the GNS gap. Our result
establishes equality after amplification by $M_2$.

We now define the complete MLSI constants by
\begin{align*}
  \alpha^{\mathrm c}(\cL)
  &:=\inf_{n\ge1}\alpha(\cL^{(n)}).
\end{align*}
The complete Pimsner--Popa index is $C_{\cb}(E):=\sup_{n\ge1}C(E^{(n)})$. Combining
Theorem~\ref{thm:intro-amplification}  with Theorem~\ref{thm:intro-mlsi} gives the
complete estimate.

\begin{corollary}[CMLSI criterion]
\label{cor:intro-cmlsi}
Under the above assumptions,
\begin{equation}
  \frac{\lambda_{\mathrm{GNS}}(\cL)}{2C_{\cb}(E)}
  \le \alpha^{\mathrm c}(\cL)
  \le \lambda_{\mathrm{GNS}}(\cL).
  \label{eq:intro-cmlsi-bounds}
\end{equation}
Consequently,
\begin{equation}
  \alpha^{\mathrm c}(\cL)>0
  \quad\Longleftrightarrow\quad
  \lambda_{\mathrm{GNS}}(\cL)>0.
  \label{eq:intro-cmlsi-criterion}
\end{equation}
\end{corollary}

The corollary follows by applying Theorem~\ref{thm:intro-mlsi} at every matrix level and
using $C(E^{(n)})\le C_{\cb}(E)$. We note that under GNS symmetry, both BKM constants
and the GNS and KMS gaps coincide with the spectral gap of $-\cL$ (see
Corollary~\ref{cor:gns-gaps}). The complete estimate above then recovers
\cite{GaoRouze2022} with no detailed balance being assumed.

For CKG samplers, the GNS Dirichlet form is an integral of squared commutator norms
involving the filtered jump operators. We show that its kernel is exactly the
fixed-point algebra, which in finite dimensions implies CMLSI (see
Section~\ref{sec:ckg}).

The paper is organized as follows. Section~\ref{sec:setup} introduces the notation and
basic assumptions. Section~\ref{sec:bkm-mlsi} proves the comparison between BKM
coercivity and MLSI. Section~\ref{sec:corners} characterizes the ordinary BKM constant
through endpoint and internal constants. Section~\ref{sec:complete-bkm} proves the
amplification theorem and treats the GNS-symmetric case. Section~\ref{sec:examples}
gives the examples separating the spectral gaps and entropy decay constants.
Section~\ref{sec:ckg} applies the CMLSI criterion to CKG samplers.

\noindent {\bf Acknowledgement. } The authors thank Zhiyan Ding and Lin Lin for helpful
discussions on relative entropy decay beyond GNS-symmetry. Li Gao and Jingyu Guo are
supported in part by the National Natural Science Foundation of China (grant no.
12401163) and the Department of Science and Technology of Hubei Province (project nos.
2025EHA041 and 2025AFA044).

\section{Preliminaries}
\label{sec:setup}

\subsection{Quantum Markov semigroups and detailed balance}
\label{subsec:qms}
Throughout, $\cH$ is finite-dimensional and $\Tr$ denotes the unnormalized trace.
\begin{definition}
\label{def:qms}
A quantum Markov semigroup (QMS) $(P_t)_{t\ge0}$ on $\B(\cH)$ is a
norm-continuous semigroup of unital completely positive maps satisfying $P_0=\id$
and $P_s\circ P_t=P_{s+t}$ for all $s,t\ge0$.
\end{definition}

The generator of the semigroup is defined
\[
  \cL(X):=\lim_{t\to 0}\frac{1}{t} (P_t(X)- X).
\]
Then $P_t=e^{t\cL}$, and the generator admits the
Gorini--Kossakowski--Lindblad--Sudarshan (GKLS) form
\cite{GoriniKossakowskiSudarshan1976,Lindblad1976}
\begin{equation}
  \cL(X)=i[H,X]
  +\sum_j\left(V_j^*XV_j-\frac12\{V_j^*V_j,X\}\right),
  \label{eq:gkls}
\end{equation}
with $H=H^*$ and $ V_j \in \B(\cH)$. Its trace predual $P_{t*}=e^{t\cL_*}$ is given by
\[
  \Tr(X P_{t*}(\rho))=\Tr( P_t(X)\rho ) , \forall X,\rho\in \B(\cH) \ .
\]

A state $\sigma$ is invariant if $P_{t*}(\sigma)=\sigma$ for every $t\ge0$, or
equivalently $\cL_*(\sigma)=0$.  For a faithful invariant state $\sigma>0$, the GNS and
KMS inner products are
\begin{equation}
  \ip{X}{Y}_{\mathrm{GNS},\sigma}:=\Tr(\sigma X^*Y),
  \qquad
  \ip{X}{Y}_{\mathrm{KMS},\sigma}
  :=\Tr(\sigma^{1/2}X^*\sigma^{1/2}Y).
  \label{eq:gns-kms-products}
\end{equation}
The QMS is called \emph{GNS symmetric} (respectively, \emph{KMS symmetric}) with respect
to $\sigma$ if $\cL$ is self-adjoint for the corresponding inner product. These
conditions are also called \emph{GNS-} and \emph{KMS-detailed balance}. GNS symmetry
implies KMS symmetry, whereas the converse need not hold \cite[Theorem~2.9 and
Appendix~B]{CarlenMaas2017}. Indeed, It is proved that a semigroup $P_t$ is GNS-symmetry
if and only if $P_t$ is KMS-symmetry and covariant under modular automorphism group
\[
  \sigma^{is}P_t(\sigma^{-is}\cdot \sigma^{is})\sigma^{-is}= P_t(\cdot) ,  \forall \ t
  \ge 0\ .
\]

Under GNS and KMS symmetry with respect to a faithful invariant state
$\sigma$, respectively, define the corresponding spectral gaps
\cite[Section~3]{WirthKMSGNS2026} by
\begin{align*}
  \lambda_{\mathrm{GNS}}(\cL)
  &:=\inf_X
    \frac{-\ip{X}{\cL(X)}_{\mathrm{GNS},\sigma}}
         {\ip{X}{X}_{\mathrm{GNS},\sigma}},
  &\qquad
  \lambda_{\mathrm{KMS}}(\cL)
  &:=\inf_X
    \frac{-\ip{X}{\cL(X)}_{\mathrm{KMS},\sigma}}
         {\ip{X}{X}_{\mathrm{KMS},\sigma}}.
\end{align*}
Here each infimum is over nonzero $X$ orthogonal to $\Ker\cL$ in the
corresponding inner product. In either case, the gap equals
$\min(\spec(-\cL)\setminus\{0\})$ and has the same value for every faithful
invariant reference state. See \cite[Lemma~3.2]{GaoRouze2022} for the GNS case.

\subsection{Conditional expectations and indices}

\begin{definition}
Let $\cN\subseteq\B(\cH)$ be a unital $*$-subalgebra.  A conditional expectation
$E:\B(\cH)\to\cN$ is a unital completely positive map onto $\cN$ satisfying $E\circ E
=E$. It is called faithful if $E(X^*X)=0$ implies $X=0$.  Its trace adjoint is denoted
by $E_*$.
\end{definition}
It is proved \cite[Theorem~1]{Tomiyama1957} that the conditional expectation satisfies
the bimodule property
\begin{equation}
  E(aXb)=aE(X)b,
  \qquad a,b\in\cN,\quad X\in\B(\cH).
  \label{eq:expectation-bimodule}
\end{equation}

\Needspace{11\baselineskip}
Throughout the paper, we will work under the following assumption.
\begin{assumption}
\label{ass:asymptotic-limit}
We assume that the limit
\begin{equation}
  E:=\lim_{t\to\infty}P_t
  \label{eq:asymptotic-expectation}
\end{equation}
exists and is a faithful conditional expectation.
\end{assumption}
Given the assumption, set the fixed-point subspace
\begin{equation}
  \Fix(P):=\{X\in\B(\cH):P_t(X)=X\ \text{for every }t\ge0\},
  \qquad \cN:=\Ran E=\Fix(P).
  \label{eq:fixed-algebra}
\end{equation}
In particular,
\begin{equation}
  P_tE=EP_t=E,
  \qquad
  P_{t*}E_*=E_*P_{t*}=E_*.
  \label{eq:asymptotic-relations}
\end{equation}

Since $E$ is unital and faithful and $P_{t*}E_*=E_*$, the state $E_*(\rho)$ is faithful
and invariant whenever $\rho$ is faithful.  In particular, a canonical choice
$\sigma_0=E_*(I/d)$, where $d=\dim\cH$.

Under GNS or KMS symmetry, $E$ is the orthogonal projection onto $\Ker\cL$
for the corresponding inner product. The orthogonality conditions defining the gaps
in Section~\ref{subsec:qms} can therefore be written as $E(X)=0$.

The QMS is called \emph{primitive} if the fixed space $\cN=\mathbb C I$ is trivial.  In
this case, $E(X)=\Tr(\sigma X)I$ is the depolarizing map for the unique faithful
invariant state $\sigma$.

\Needspace{10\baselineskip}
We recall the standard form of a faithful conditional expectation \cite[Section~2.1,
Eq.~(13)]{GaoRouze2022}.

\begin{proposition}[Structure of conditional expectations]
\label{prop:expectation-structure}
Let $\cH$ be finite-dimensional, $\cN\subseteq\B(\cH)$ be a unital $*$-subalgebra, and
let $E:\B(\cH)\to\cN$ be a faithful conditional expectation.  Up to unitary equivalence,
there exist a finite index set $\mathcal I$ and finite-dimensional Hilbert spaces
$\cK_k,\cR_k, k\in \mathcal{I}$ such that
\begin{equation}
  \cH=\bigoplus_{k\in\mathcal I}\cK_k\otimes\cR_k,
  \qquad
  \cN=\bigoplus_{k\in\mathcal I}
  \B(\cK_k)\otimes I_{\cR_k}.
  \label{eq:wedderburn}
\end{equation}
Let $p_k$ denote the projection onto the $k$th summand, and write $X_{kk}:=p_kXp_k$ and
$\rho_{kk}:=p_k\rho p_k$.  There are faithful density matrices $\sigma_k\in\B(\cR_k)$
for which
\begin{align}
  E(X)
  &=\bigoplus_{k\in\mathcal I}
  \left[(\id_{\cK_k}\otimes\varphi_{\sigma_k})(X_{kk})\right]\otimes I_{\cR_k},
  \label{eq:expectation-blocks}
  \\
  E_*(\rho)
  &=\bigoplus_{k\in\mathcal I}
  \Tr_{\cR_k}(\rho_{kk})\otimes\sigma_k,
  \label{eq:predual-blocks}
\end{align}
where $\varphi_{\sigma_k}(Z):=\Tr(\sigma_k Z)$.  The states fixed by $E_*$ are precisely
those of the form
\begin{equation}
  \rho=\bigoplus_{k\in\mathcal I}\eta_k\otimes\sigma_k,
  \qquad \eta_k\ge0,
  \qquad \sum_{k\in\mathcal I}\Tr\eta_k=1,
  \label{eq:fixed-states}
\end{equation}
where $\eta_k\in \mathcal{B}(\cK_k)$. Such an invariant state $\rho$ is faithful if and
only if $\eta_k$ is strictly positive for every $k$.
\end{proposition}

Let $\cD(\cH)$ denote the set of states (density operators) on $\cH$, and set the set of
invariant states as
\begin{equation}
  \cD(E):=\{\sigma\in\cD(\cH):E_*(\sigma)=\sigma\}.
  \label{eq:invariant-states}
\end{equation}

\begin{definition}
\label{def:pp-indices}
The ordinary and complete Pimsner--Popa indices \cite{PimsnerPopa1986,GaoRouze2022} of a
faithful conditional expectation $E$ are defined by
\begin{align}
  C(E)
  &:=\inf\{C>0:\rho\le CE_*\rho\ \text{for every }\rho\ge0\},
  \label{eq:pp-index}
  \\
  C_{\cb}(E)
  &:=\sup_{n\ge1}C(\id_{M_n}\otimes E).
  \label{eq:complete-pp-index}
\end{align}
\end{definition}
We have the following formula for Pimsner--Popa indices in finite dimensions.
\begin{proposition}[{\cite[Proposition~3.2]
{GaoZhao2026}}]
\label{prop:pp-indices}
Let $E$ be a faithful conditional expectation in the form of Proposition
\ref{prop:expectation-structure}. Let $m_k=\dim\cK_k$ and $d_k=\dim\cR_k$, and list the
eigenvalues  of $\sigma_k^{-1}$ in decreasing order as
$\lambda_1^\downarrow(\sigma_k^{-1}), \cdots, \lambda_{d_k}^\downarrow(\sigma_k^{-1})$,
counted with multiplicity.  Then
\begin{equation}
  C(E)=\sum_{k\in\mathcal I}
  \sum_{j=1}^{\min\{m_k,d_k\}}
  \lambda_j^\downarrow(\sigma_k^{-1}),
  \qquad
  C_{\cb}(E)=\sum_{k\in\mathcal I}\Tr(\sigma_k^{-1}).
  \label{eq:pp-index-formulas}
\end{equation}
\end{proposition}
Via trace duality, the Pimsner--Popa indices are equivalently formulated in the
Heisenberg picture
\begin{align}
  C(E)
  &=\inf\{C>0:X\le CE(X)\ \text{for all }X\ge0\},
  \label{eq:pp-order}
  \\
  C_{\cb}(E)
  &=\inf\{C>0:CE-\id\ \text{is completely positive}\}.
  \label{eq:pp-cp-order}
\end{align}

\begin{example}
The ordinary and complete indices need not coincide.  For $\cH=\cK\otimes\mathbb C^d$,
consider the conditional expectation
$E(X)=[(\id_{\B(\cK)}\otimes\frac{\Tr}{d})(X)]\otimes I_d$.  Then
\[
  C(E)=d\min\{\dim\cK,d\},
  \qquad C_{\cb}(E)=d^2.
\]
For the replacer channel $E_*(\rho)=\Tr(\rho)\sigma$, where $\sigma$ is a faithful
state, the indices are
\[
  C(E)=\norm{\sigma^{-1}}_\infty,
  \qquad C_{\cb}(E)=\Tr(\sigma^{-1}).
\]
\end{example}

\subsection{Modified log-Sobolev constant}

For states $\rho$ and $\sigma$, define the Umegaki relative entropy
\begin{equation}
  D(\rho\Vert\sigma):=\Tr\rho(\log\rho-\log\sigma),
  \label{eq:relative-entropy}
\end{equation}
given the support condition $\text{supp}(\rho)\subset \text{supp}(\sigma)$. The
definition extends to the pair $\rho$ and $\sigma$ with equal trace.  Relative to the
conditional expectation E $E$, define
\begin{equation}
  D_{E}(\rho):=\inf_{\sigma\in\cD(E)}D(\rho\Vert\sigma)=D(\rho\Vert E_*\rho).
  \label{eq:conditional-entropy}
\end{equation}
Here the infimum in \eqref{eq:conditional-entropy} is attained at $E_*\rho$ by the chain
rule \cite[Section~2.1, Eq.~(12)]{GaoRouze2022} is
\begin{equation}
  D(\rho\Vert\omega)
  =D(\rho\Vert E_*\rho)+D(E_*\rho\Vert\omega),
  \qquad \omega\in\cD(E),
  \label{eq:pythagorean}
\end{equation}
In particular, $D(\rho\Vert E_*\rho)<\infty$ because $\rho\le C(E)E_*\rho$.

For a faithful state $\rho$, define the entropy production \cite{Spohn1978}
\begin{equation}
  \EP_{\cL}(\rho)
  :=-\Tr\!\left[\cL_*(\rho)(\log\rho-\log E_*\rho)\right]
  =-\left.\frac{\dd}{\dd t}\right|_{t=0^+}
  D(P_{t*}\rho\Vert E_*\rho).
  \label{eq:entropy-production}
\end{equation}
The entropy production $\EP_{\cL}$ is always nonnegative by data processing inequality,
since $P_{t*}E_*=E_*$.

\begin{definition}
\label{def:mlsi}
The modified log-Sobolev inequality (MLSI) constant of $\cL$ relative to $\cN$ is
\begin{equation}
  \alpha(\cL)
  :=\inf_{\substack{\rho\in\cD(\cH),\ \rho>0\\
  \rho\ne E_*\rho}}
  \frac{\EP_{\cL}(\rho)}{2D(\rho\Vert E_*\rho)}.
  \label{eq:mlsi-constant}
\end{equation}
\end{definition}
Equivalently, $\alpha(\cL)\ge\alpha$ means
\begin{equation}
  \EP_{\cL}(\rho)\ge2\alpha D(\rho\Vert E_*\rho)
  \label{eq:mlsi-ineq}
\end{equation}
for every faithful state. By standard Gronwall Lemma, \eqref{eq:mlsi-ineq} is equivalent
to
\begin{equation}
  D(P_{t*}\rho\Vert E_*\rho)
  \le e^{-2\alpha t}D(\rho\Vert E_*\rho),
  \qquad t\ge0,
  \label{eq:mlsi-decay}
\end{equation}
for every faithful state $\rho$. For $\alpha>0$, \eqref{eq:mlsi-decay} implies $P_t\to
E$. For MLSI in the non-primitive setting was introduced in Bardet
\cite{Bardet2017Decoherence} and studied in
\cite{GaoRouze2022,GaoJungeLaRacuenteLi2025}. We will use the normalization
$\EP\ge2\alpha D$ throughout.

For $n\ge1$, set the matrix level amplification
\[
  \cL^{(n)}:=\id_{M_n}\otimes\cL,
  \qquad
  E^{(n)}:=\id_{M_n}\otimes E,
\]
and let $\alpha(\cL^{(n)})$ denote the corresponding MLSI constant of $\cL^{(n)}$
relative to $\Ran E^{(n)}=M_n\otimes\cN$.  The complete MLSI constant is defined as
\begin{equation}
  \alpha^{\mathrm c}(\cL)
  :=\inf_{n\ge1}\alpha(\cL^{(n)}).
  \label{eq:cmlsi-constant}
\end{equation}
Namely, $\alpha^{\mathrm c}(\cL)\ge\alpha$ if and only if \eqref{eq:mlsi-decay} holds
for every finite matrix amplification.  The complete MLSI constant was introduced in
\cite{GaoJungeLaRacuente2020fisher} for the tensorization property that for two quantum
Markov semigroups $P_t=e^{\cL_1 t}$ and $Q_t=e^{\cL_2 t}$,
\[
  \alpha^{\mathrm c}(\cL_1\otimes  \id +\id \otimes \cL_2)=\min\{ \alpha^{\mathrm
  c}(\cL_1) ,\alpha^{\mathrm c}(\cL_2)\}\ .
\]

\subsection{BKM metric and coercivity constants}
\label{subsec:bkm-metric}
For $\sigma>0$, define the Bogoliubov--Kubo--Mori (BKM) map
\cite{Petz1996,LesniewskiRuskai1999}
\begin{equation}
  \Gamma_\sigma(X):=\int_0^1\sigma^sX\sigma^{1-s}\dd s.
  \label{eq:bkm-map}
\end{equation}
Given the spectrum decomposition $\sigma=\sum_i\lambda_i|i\rangle\langle i|$ and
$E_{ij}:=|i\rangle\langle j|$, then
\begin{equation}
  \Gamma_\sigma(E_{ij})=m(\lambda_i,\lambda_j)E_{ij},
  \qquad
  m(a,b)=\frac{a-b}{\log a-\log b}
  \label{eq:logarithmic-mean}
\end{equation}
with $m(a,a)=a$ by continuity. The inverse of $\Gamma_\sigma$ defines the Hessian metric
of relative entropy.

\begin{lemma}
\label{lem:bkm-calculus}
For every $\rho>0$ and $X=X^*\in\B(\cH)$, the Fr\'echet derivative of the matrix
logarithm satisfies \cite[Section~III]{Lieb1973}
\begin{align}
  D(\log)_\rho[X]
  &:=\left.\frac{\dd}{\dd u}\right|_{u=0}\log(\rho+uX) \nonumber\\
  &=\int_0^\infty(\rho+sI)^{-1}X(\rho+sI)^{-1}\dd s
  =\Gamma_\rho^{-1}(X).
  \label{eq:frechet-log}
\end{align}
For strictly positive states $\rho,\sigma$, set $A:=\rho-\sigma$ and
$\rho_t:=t\rho+(1-t)\sigma$. Consequently,
\begin{align}
  \log\rho-\log\sigma
  &=\int_0^1\Gamma_{\rho_s}^{-1}(A)\dd s.
  \label{eq:log-radial}
\end{align}
Moreover, the relative entropy admits the integral representation
\cite[Lemma~2.2]{GaoRouze2022}
\begin{align}
  D(\rho\Vert\sigma)
  &=\int_0^1(1-t)\Tr[A\Gamma_{\rho_t}^{-1}(A)]\dd t.
  \label{eq:entropy-integral}
\end{align}
In particular, for every strictly positive state $\sigma$ and traceless self-adjoint
$A$,
\begin{equation}
  D(\sigma+\eps A\Vert\sigma)
  =\frac{\eps^2}{2}\Tr[A\Gamma_\sigma^{-1}(A)]+o(\eps^2)
  \qquad (\eps\to0).
  \label{eq:entropy-hessian}
\end{equation}
\end{lemma}

Let $\sigma$ be a faithful invariant state and let $A=A^*$ satisfy $E_*A=0$. Since $E_*$
is trace preserving, $\Tr A=0$.  Thus, for all sufficiently small real $\eps$,
$\sigma+\eps A$ is a faithful state with $E_*(\sigma+\eps A)=\sigma$.  Expanding
\eqref{eq:entropy-production} gives
\begin{equation}
  \EP_{\cL}(\sigma+\eps A)
  =-\eps^2\Tr[\cL_*(A)\Gamma_\sigma^{-1}(A)]+o(\eps^2).
  \label{eq:production-hessian}
\end{equation}

\Needspace{12\baselineskip}
\begin{definition}
\label{def:bkm-constants}
The \emph{BKM coercivity constant} of $\cL$ is
\begin{equation}
  \lambda_{\mathrm{BKM}}(\cL)
  :=\inf_{\substack{\sigma>0,\ \Tr\sigma=1,\ \cL_*(\sigma)=0\\
  A=A^*\ne0,\ E_*A=0}}
  \frac{-\Tr[\cL_*(A)\Gamma_\sigma^{-1}(A)]}
  {\Tr[A\Gamma_\sigma^{-1}(A)]},
  \label{eq:bkm-constant}
\end{equation}
The infimum ranges over both faithful invariant states $\sigma$ and centered
self-adjoint directions $A$. For $n\ge1$, write $\lambda_{\mathrm{BKM}}^{(n)}(\cL)
:=\lambda_{\mathrm{BKM}}(\id_{M_n}\otimes\cL)$. The \emph{complete BKM coercivity
constant} is
\begin{equation}
  \lambda_{\mathrm{BKM}}^{\mathrm c}(\cL)
  :=\inf_{n\ge1}\lambda_{\mathrm{BKM}}^{(n)}(\cL).
  \label{eq:complete-bkm-constant}
\end{equation}
\end{definition}
Linearization gives the following upper bounds.
\begin{proposition}
\label{prop:alpha-upper-bkm}
Let $P_t=e^{t\cL}$ be a QMS with a faithful asymptotic conditional expectation $E$. Then
\begin{equation}
  \alpha(\cL)\le\lambda_{\mathrm{BKM}}(\cL),
  \qquad
  \alpha^{\mathrm c}(\cL)\le\lambda_{\mathrm{BKM}}^{\mathrm c}(\cL).
  \label{eq:mlsi-upper-bounds}
\end{equation}
\end{proposition}

\begin{proof}
Fix an admissible pair $(\sigma,A)$ in the infimum of
Definition~\ref{def:bkm-constants}. For every sufficiently small nonzero real $\eps$,
the state $\rho_\eps:=\sigma+\eps A$ is faithful and $E_*\rho_\eps=\sigma$.  By
\eqref{eq:entropy-hessian} and \eqref{eq:production-hessian},
\[
  \lim_{\eps\to0}
  \frac{\EP_{\cL}(\rho_\eps)}
  {2D(\rho_\eps\Vert E_*\rho_\eps)}
  =\frac{-\Tr[\cL_*(A)\Gamma_\sigma^{-1}(A)]}
  {\Tr[A\Gamma_\sigma^{-1}(A)]}.
\]
Taking the infimum over admissible $(\sigma,A)$ proves the bound. Applying the same
argument to $\id_{M_n}\otimes\cL$ for every $n$, and then taking the infimum over $n$,
proves the complete bound.
\end{proof}

\section{MLSI bound via BKM coercivity}
\label{sec:bkm-mlsi}
We derive lower bounds for the MLSI constants in terms of BKM coercivity and the
Pimsner--Popa indices.  The proof uses Frenkel's formula to express the relative entropy
and BKM metric contraction deficits as integrals against the same positive measure.

For $\sigma>0$ and $A=A^*$, write
\begin{equation}
  H_\sigma(A):=\Tr[A\Gamma_\sigma^{-1}(A)],
  \qquad
  Q_{\sigma,\cL}(A):=-\Tr[\cL_*(A)\Gamma_\sigma^{-1}(A)],
  \label{eq:bkm-forms}
\end{equation}
and abbreviate $Q_{\sigma,\cL}$ to $Q_\sigma$ whenever the generator is clear.

\subsection{An entropy-metric comparison}
For states $\rho,\sigma$ with $\sigma>0$, $H_\sigma(\rho-\sigma)$ is the quantum
$\chi^2$-divergence associated with BKM metric; see \cite[Definition~1 and
Eq.~(54)]{TemmeKastoryanoRuskaiWolfVerstraete2010}. We recall the entropy-metric
comparison from \cite[Lemma~2.2]{GaoRouze2022} (see also
\cite{LesniewskiRuskai1999,TemmeKastoryanoRuskaiWolfVerstraete2010}).
\begin{lemma}
\label{lem:entropy-metric}
Let $\rho\ge0$ and $\sigma>0$ have equal trace and $\rho\le C\sigma$.  Then
\begin{equation}
  D(\rho\Vert\sigma)\le H_\sigma(\rho-\sigma).
  \label{eq:entropy-metric}
\end{equation}
\end{lemma}

We present a comparison lemma for the contraction difference in terms of relative
entropy and BKM metric. It is an application of the Frenkel's integral formula of
relative entropy \cite[Theorem~6]{Frenkel2023}
\begin{equation}
  D(\rho\Vert\sigma)
  =\int_{-\infty}^{\infty}
  \frac{\Tr^-\!\left((1-u)\rho+u\sigma\right)}
  {|u|(u-1)^2}\dd u.
  \label{eq:frenkel-entropy}
\end{equation}
Here, for a Hermitian operator $X$, we denote
\[
  \Tr^-(X):=\Tr(X_-),
  \qquad
  X_-:=\frac{|X|-X}{2}.
\]

\begin{lemma}
\label{lem:entropy-deficit}
Let $\rho\ge0$ and $\sigma>0$ have equal trace, and let $T_*$ be a positive
trace-preserving map satisfying $T_*(\sigma)=\sigma$.  If $\rho\le C\sigma$ for some
$C>1$, then
\begin{equation}
  D(\rho\Vert\sigma)-D(T_*\rho\Vert\sigma)
  \ge\frac1{2C}
  \left[H_\sigma(\rho-\sigma)-H_\sigma(T_*\rho-T_*\sigma)\right].
  \label{eq:deficit-comparison}
\end{equation}
\end{lemma}
\begin{proof}
Set $A=\rho-\sigma$. The assumptions give
\[
  \Tr A=0,
  \qquad
  -\sigma\le A\le(C-1)\sigma.
\]
For every real $r$ for which $\rho_r:=\sigma+rA\ge0$, define
\begin{equation}
  \Delta_T(r)
  :=D(\rho_r\Vert\sigma)-D(T_*\rho_r\Vert\sigma).
  \label{eq:entropy-deficit}
\end{equation}
We also set
\begin{equation}
  \delta_T(z)
  :=\Tr^-(\sigma+zA)
  -\Tr^-\!\left(T_*(\sigma+zA)\right),
  \qquad z\in\mathbb R.
  \label{eq:negative-part-deficit}
\end{equation}
The quantity $\delta_T(z)$ is the decrease in the trace of the negative part of
$\sigma+zA$ under $T_*$, which is always non-negative by the monotonicity $
 \Tr^-(T_*X)\le\Tr^-(X)
$ for positive and trace preserving $T_*$ and self-adjoint $X^*=X$.

\noindent {\bf Step 1: the defect measure.} For $0\le z\le1$,
\[
  \sigma+zA=(1-z)\sigma+z\rho\ge0.
\]
For $-1/(C-1)\le z\le0$, the upper bound $A\le(C-1)\sigma$ gives
\[
  \sigma+zA\ge\bigl(1+z(C-1)\bigr)\sigma\ge0.
\]
Positivity of $T_*$ gives $T_*(\sigma+zA)\ge0$ throughout $[-1/(C-1),1]$.  Consequently,
\begin{equation}
  \delta_T(z)=0
  \qquad \text{for}\qquad
  -\frac1{C-1}\le z\le1.
  \label{eq:deficit-support}
\end{equation}
Moreover, $X\mapsto\Tr^-(X)$ is trace-norm continuous and
\[
  0\le\delta_T(z)
  \le\Tr^-(\sigma+zA)
  \le\norm{\sigma}_1+|z|\norm{A}_1.
\]
Define
\begin{equation}
  \dd\nu_T(z):=\frac{\delta_T(z)}{|z|^3}\dd z,
  \label{eq:deficit-measure}
\end{equation}
where the density is set to zero at $z=0$. Note that the Eq. \eqref{eq:deficit-support}
removes the singularity at the origin, while the linear growth bound makes the density
$O(|z|^{-2})$ at infinity. Thus $\nu_T$ is a finite positive measure.

\noindent {\bf Step 2: The entropy difference.} Apply \eqref{eq:frenkel-entropy} to
$\rho_r=\sigma+rA$ and $T_*\rho_r=\sigma+rT_*A$.  Since $\sigma>0$, both relative
entropies and both Frenkel integrals are finite.  Set $z=r(1-u)$.  Using
$T_*(\sigma)=\sigma$, one has
\begin{align*}
  (1-u)\rho_r+u\sigma
  &=\sigma+zA,\\
  (1-u)T_*\rho_r+u\sigma
  &=T_*(\sigma+zA),\\
  |u|(u-1)^2
  &=\frac{|z-r|z^2}{|r|^3},
  \qquad
  |\dd u|=\frac{|\dd z|}{|r|}.
\end{align*}
Taking the absolute value of the Jacobian gives, for either sign of $r$,
\begin{equation}
  \Delta_T(r)=D(\rho_r\Vert\sigma)-D(T_*\rho_r\Vert\sigma)
  =r^2\int_{-\infty}^{\infty}
  \frac{\delta_T(z)}{z^2|z-r|}\dd z.
  \label{eq:deficit-integral}
\end{equation}
To check local integrability at $z=r$, use $\rho_r\ge0$ and the variational formula for
$\Tr^-$ to obtain
\[
  \Tr^-(\sigma+zA)
  \le\norm{(\sigma+zA)-\rho_r}_1
  =|z-r|\norm{A}_1.
  \] Equivalently, using the measure  $\nu_T$, we have
  \begin{equation}
    \frac{\Delta_T(r)}{r^2}
    =\int_{\mathbb R}\frac{|z|}{|z-r|}\dd\nu_T(z).
    \label{eq:normalized-deficit}
  \end{equation}
  \noindent{\bf Step 3: The total mass of $\nu_T$.}
  Set
  \[
  \eta:=\min\left\{1,\frac1{C-1}\right\}>0.
\]
By \eqref{eq:deficit-support}, $\nu_T([-\eta,\eta])=0$.  Hence, for $|r|\le\eta/2$,
\[
  \frac{|z|}{|z-r|}\le2
  \qquad\text{for $\nu_T$-almost every $z$}.
\]
The kernel converges to $1$ for $\nu_T$-almost every $z$ as $r\to0$. Dominated
convergence in \eqref{eq:normalized-deficit} therefore gives
\begin{equation}
  \lim_{r\to0}\frac{\Delta_T(r)}{r^2}
  =\nu_T(\mathbb R).
  \label{eq:measure-mass-limit}
\end{equation}
A positive map preserves Hermitian operators, and trace preservation gives
$\Tr(T_*A)=0$.  For all sufficiently small real $r$, both $\sigma+rA$ and $\sigma+rT_*A$
are positive definite.  By homogeneity, the BKM expansion \eqref{eq:entropy-hessian}
gives
\[
  \Delta_T(r)
  =\frac{r^2}{2}
  \left[H_\sigma(A)-H_\sigma(T_*A)\right]+o(r^2).
\]
Comparing with \eqref{eq:measure-mass-limit}, we obtain the exact identity
\begin{equation}
  \nu_T(\mathbb R)
  =\frac12\left[H_\sigma(A)-H_\sigma(T_*A)\right].
  \label{eq:measure-bkm-mass}
\end{equation}
Thus the BKM contraction deficit is twice the total mass of $\nu_T$. We note that the
same identity essentially follows by subtracting Frenkel's Hessian formulas
\cite[Theorem~7(a), $m=2$]{Frenkel2023} for $A$ and $T_*A$ at $\sigma$.

\noindent {\bf Step 4: Compare the two kernels.} Taking $r=1$ in
\eqref{eq:normalized-deficit} gives
\begin{equation}
  D(\rho\Vert\sigma)-D(T_*\rho\Vert\sigma)
  =\int_{\mathbb R}\frac{|z|}{|z-1|}\dd\nu_T(z).
  \label{eq:entropy-deficit-mass}
\end{equation}
By \eqref{eq:deficit-support}, $\nu_T$ is concentrated on
\[
  \left(-\infty,-\frac1{C-1}\right)\cup(1,\infty).
\]
The claim then follows from the elementary inequality
\begin{align*}
  \frac{|z|}{|z-1|}
  &\ge\frac1C,
  \qquad
  z>1\ \text{or}\ z<-\frac1{C-1}.
  \qedhere
\end{align*}
\end{proof}
\subsection{From tangent coercivity to entropy decay}
We now apply the above lemma to $T_*=P_{s*}$ and differentiate at $s=0$.
\begin{theorem}
\label{thm:entropy-production}
Let $P_t=e^{t\cL}$ be a QMS with a faithful asymptotic conditional expectation $E$. Then
every faithful state $\rho\in\cD(\cH)$ satisfies
\begin{equation}
  \EP_{\cL}(\rho)\ge\frac{1}{C(E)} Q_{E_*\rho}(\rho-E_*\rho).
  \label{eq:production-bkm}
\end{equation}
Consequently,
\begin{equation}
  \alpha(\cL)\ge\frac{\lambda_{\mathrm{BKM}}(\cL)}{2{C(E)}}.
  \label{eq:mlsi-lower-bound}
\end{equation}
\end{theorem}

\begin{proof}
Fix a faithful state $\rho$ and set $\sigma=E_*\rho$, $A=\rho-\sigma$, and $C=C(E)$. The
state $\sigma$ is faithful and invariant. The index bound gives $\rho\le C\sigma$. For
$s>0$, applying Lemma~\ref{lem:entropy-deficit} to $T_*=P_{s*}$ gives
\begin{equation}
  \frac{D(\rho\Vert\sigma)-D(P_{s*}\rho\Vert\sigma)}s
  \ge\frac1{2C}
  \frac{H_\sigma(A)-H_\sigma(P_{s*}A)}s.
  \label{eq:semigroup-deficit}
\end{equation}
As $s\downarrow0$, the left-hand side converges to $\EP_{\cL}(\rho)$, whereas
self-adjointness of $\Gamma_\sigma^{-1}$ for the Hilbert--Schmidt inner product gives
\begin{equation}
  -\frac12\left.\frac{\dd}{\dd s}H_\sigma(P_{s*}A)\right|_{s=0}
  =-\Tr[\cL_*(A)\Gamma_\sigma^{-1}(A)]
  =Q_\sigma(A).
  \label{eq:bkm-derivative}
\end{equation}
Hence $\EP_{\cL}(\rho)\ge C^{-1}Q_\sigma(A)$.  Lemma~\ref{lem:entropy-metric} and
definition of $\lambda_{\mathrm{BKM}}(\cL)$ then yield
\[
  \EP_{\cL}(\rho)
  \ge\frac1C Q_\sigma(A)
  \ge\frac{\lambda_{\mathrm{BKM}}(\cL)}{C} H_\sigma(A)
  \ge\frac{\lambda_{\mathrm{BKM}}(\cL)}{C} D(\rho\Vert\sigma),
\]
which is \eqref{eq:mlsi-lower-bound} in the normalization $\EP\ge2\alpha D$.
\end{proof}
Applying the above theorem to $\cL^{(n)}=\id_{M_n}\otimes\cL$ with $C=C(\id_{M_n}\otimes
E)$ for every $n\ge1$ gives the complete MLSI bound.
\begin{corollary}
\label{cor:bkm-mlsi}
Let $P_t=e^{t\cL}$ be a QMS with a faithful asymptotic conditional expectation $E$. Then
\begin{equation}
  \frac{\lambda_{\mathrm{BKM}}(\cL)}{2C(E)}
  \le \alpha(\cL)
  \le \lambda_{\mathrm{BKM}}(\cL),
  \label{eq:mlsi-bounds}
\end{equation}
and
\begin{equation}
  \frac{\lambda_{\mathrm{BKM}}^{\mathrm c}(\cL)}{2C_{\cb}(E)}
  \le \alpha^{\mathrm c}(\cL)
  \le \lambda_{\mathrm{BKM}}^{\mathrm c}(\cL).
  \label{eq:cmlsi-bkm-bounds}
\end{equation}
Hence positivity of the BKM coercivity constant is equivalent to positivity of the
corresponding MLSI constant, for both the ordinary and complete constants.
\end{corollary}
The above corollary extends \cite[Theorem~3.3]{GaoRouze2022} beyond GNS symmetry.
Indeed,
we will show in Corollary~\ref{cor:gns-gaps} that $\lambda_{\mathrm{BKM}}$ coincides
with both the KMS and GNS spectral gaps under GNS symmetry.

\section{Corner decomposition and BKM coercivity}
\label{sec:corners}
In this section, we restrict the generator to the rectangular corners of the fixed-point
decomposition and express the BKM forms in these coordinates. Endpoint interpolation
then yields an exact expression for the ordinary BKM coercivity constant. We use the
notation $H_\omega$ and $Q_\omega$ from \eqref{eq:bkm-forms}.

\subsection{Fixed-point structure and corner generators}
\label{subsec:corner-generators}
Let $P_t=e^{t\cL}$ be a QMS with a faithful asymptotic conditional expectation $E$. Set
$\cN=\Ran E$. Every $a\in\cN$ belongs to the multiplicative domain of $P_t$, so $P_t$ is
$\cN$-bimodular. Differentiation gives the same property for $\cL$. Thus, for every
$a,b\in\cN$, $X\in\B(\cH)$, and $t\ge0$,
\begin{equation}
  P_t(aXb)=aP_t(X)b,
  \qquad
  \cL(aXb)=a\cL(X)b.
  \label{eq:semigroup-bimodule}
\end{equation}
Such an expectation exists whenever $P_t$ is KMS symmetric with respect to a faithful
invariant state $\sigma$. Indeed, $\cL$ is diagonalizable with spectrum in
$(-\infty,0]$, so $P_t=e^{t\cL}$ converges to the projection onto $\Ker\cL$.

By the orthogonal decomposition in Proposition~\ref{prop:expectation-structure}, we may
write
\begin{align}
  \cH=\bigoplus_{k\in\mathcal I}\cK_k\otimes\cR_k,\qquad
  \cN=\bigoplus_{k\in\mathcal I}\B(\cK_k)\otimes I_{\cR_k},\qquad
  \omega=\bigoplus_{k\in\mathcal I}\eta_k\otimes\sigma_k.
  \label{eq:invariant-blocks}
\end{align}
Here the faithful density matrices $\sigma_k$ are determined by $E$, while a faithful
invariant state $\omega$ is specified by
\[
  \eta_k\in\B(\cK_k),
  \qquad \eta_k>0,
  \qquad \sum_k\Tr\eta_k=1.
\]

Let $p_k$ be the projection onto $\cK_k\otimes\cR_k$, and set
\[
  \mathsf C_{kl}:=p_k\B(\cH)p_l
  =\B(\cK_l,\cK_k)\otimes\B(\cR_l,\cR_k).
\]
Since $p_k,p_l\in\cN$, the bimodularity makes each $\mathsf C_{kl}$ invariant under
$\cL$.  Fix orthonormal bases of the noiseless factors $\B(\cK_l,\cK_k)$ and write the
matrix entry as
\[
  E_{ij}^{kl}:=|k,i\rangle\langle l,j|
  \in\B(\cK_l,\cK_k),
  \qquad (E_{ij}^{kl})^*=E_{ji}^{lk}.
\]
For each $(k,l)$, the bimodule property $\cL$ gives a unique linear map
\begin{equation}
  L_{kl}:\B(\cR_l,\cR_k)\longrightarrow\B(\cR_l,\cR_k)
  \label{eq:corner-generator}
\end{equation}
such that
\begin{equation}
  \cL(A\otimes Z)=A\otimes L_{kl}(Z)
  \qquad
  A\in\B(\cK_l,\cK_k),\ Z\in\B(\cR_l,\cR_k).
  \label{eq:corner-action}
\end{equation}
Indeed, bimodule property restricts $\cL(E_{ij}^{kl}\otimes Z)=E_{ij}^{kl}\otimes
L_{kl}(Z)$ because of the multiplication by $E_{i1}^{kk}$ on the left and $E_{1j}^{ll}$
on the right, and hence for every $A\in\B(\cK_l,\cK_k)$.  This determines $L_{kl}$
independently of the chosen bases.  Since $\cL(X^*)=\cL(X)^*$,
\begin{equation}
  L_{lk}(Z^*)=L_{kl}(Z)^*.
  \label{eq:corner-adjoints}
\end{equation}

We call $\B(\cR_l,\cR_k)$ a directed corner, a cross-corner when $k\ne l$, and a
self-corner when $k=l$, and we write $S_t^{kl}:=e^{tL_{kl}}$ as the restriction on the
$(k,l)$ corner. Since $\Ker\cL=\cN$,
\begin{equation}
  \Ker L_{kl}=\{0\}\quad(k\ne l),
  \qquad
  \Ker L_{kk}=\mathbb C I_{\cR_k}.
  \label{eq:corner-kernels}
\end{equation}
Moreover, $E\cL=0$ gives $\Tr(\sigma_kL_{kk}(Z))=0$ for every $Z\in \B(\cR_k)$. Define
\begin{equation}
  V_{kl}:=
  \begin{cases}
    \B(\cR_l,\cR_k),&k\ne l,\\[
    1mm]
    \{Z \in \B(\cR_k):\Tr(\sigma_kZ)=0\},&k=l.
  \end{cases}
  \label{eq:centered-corners}
\end{equation}
We call $V_{kl}$ the \emph{centered corner space}.  Each $V_{kl}$ is invariant under
$S_t^{kl}$.

For $Z,W\in\B(\cR_l,\cR_k)$, set the corner KMS inner product as
\begin{equation}
  h_{kl,1/2}(Z,W)
  :=\Tr\!\left(Z^*\sigma_k^{1/2}W\sigma_l^{1/2}\right).
  \label{eq:corner-kms-product}
\end{equation}
For every faithful invariant state $\omega$, the KMS inner product on $\mathsf C_{kl}$
factorizes as
\[
  \ip{A\otimes Z}{B\otimes W}_{\mathrm{KMS},\omega}
  =\Tr(\eta_l^{1/2}A^*\eta_k^{1/2}B)\,h_{kl,1/2}(Z,W).
\]
The first factor is a positive definite inner product on $\B(\cK_l,\cK_k)$.

\subsection{Weighted forms and endpoint interpolation}
\label{subsec:interpolation}

For $0\le s\le1$ and $Z,W\in\B(\cR_l,\cR_k)$, define
\begin{align}
  h_{kl,s}(Z,W)
  &:=\Tr\!\left(Z^*\sigma_k^sW\sigma_l^{1-s}\right),
  \qquad h_{kl,s}(Z):=h_{kl,s}(Z,Z),
  \label{eq:h-s}
  \\
  q_{kl,s}(Z)
  &:=-\Re\Tr\!\left(Z^*\sigma_k^sL_{kl}(Z)\sigma_l^{1-s}\right).
  \label{eq:q-s}
\end{align}
At $s=\tfrac12$, $h_{kl,s}$ is the corner KMS inner product from
Section~\ref{subsec:corner-generators}.  At the endpoints,
\begin{align*}
  h_{kl,0}(Z)&=\Tr(Z^*Z\sigma_l),&
  q_{kl,0}(Z)&=-\Re\Tr(Z^*L_{kl}(Z)\sigma_l),\\
  h_{kl,1}(Z)&=\Tr(ZZ^*\sigma_k ),&
  q_{kl,1}(Z)&=-\Re\Tr(L_{kl}(Z)Z^*\sigma_k).
\end{align*}
Kadison's inequality implies that $S_t^{kl}$ is contractive
\begin{align}
  \label{eq:endpoint-contraction}
  h_{kl,0}(S_t^{kl}Z)\le h_{kl,0}(Z),\qquad  h_{kl,1}(S_t^{kl}Z)\le h_{kl,1}(Z),
\end{align}
in both endpoint norms.  Indeed, for $Y=E_{ij}^{kl}\otimes Z$,
\[
  P_t(Y)^*P_t(Y)\le P_t(Y^*Y),
  \qquad
  P_t(Y)P_t(Y)^*\le P_t(YY^*).
\]
Taking expectations in a faithful invariant state $\omega$ and cancelling the positive
factors from $\eta_l$ and $\eta_k$ yields \eqref{eq:endpoint-contraction}.
Differentiation at $t=0$ yields
\begin{equation}
  q_{kl,0}(Z)\ge0,
  \qquad
  q_{kl,1}(Z)\ge0.
  \label{eq:endpoint-nonnegative}
\end{equation}

\begin{definition}
\label{def:endpoint-constants}
For $a\in\{0,1\}$, set
\begin{equation}
  \kappa_{kl}^{(a)}
  :=\inf_{0\ne Z\in V_{kl}}
  \frac{q_{kl,a}(Z)}{h_{kl,a}(Z)},
  \label{eq:endpoint-constants}
\end{equation}
with value $+\infty$ when $V_{kl}=\{0\}$.
\end{definition}
Fix a corner $(k,l)$ and write $h_s=h_{kl,s}$, $q_s=q_{kl,s}$, $L=L_{kl}$, and
$S_t=e^{tL}$.

\begin{lemma}
\label{lem:endpoint-interpolation}
Suppose that $V_{kl}\ne\{0\}$, and let $\kappa_0,\kappa_1\in\mathbb R$ satisfy
\begin{equation}
  q_0(Z)\ge\kappa_0h_0(Z),
  \qquad
  q_1(Z)\ge\kappa_1h_1(Z)
  \qquad \text{ for } Z\in V_{kl}.
  \label{eq:endpoint-bounds}
\end{equation}
Then, for every $s\in[0,1]$ and $Z\in V_{kl}$,
\begin{equation}
  q_s(Z)\ge\bigl((1-s)\kappa_0+s\kappa_1\bigr)h_s(Z).
  \label{eq:interpolated-form}
\end{equation}
\end{lemma}

\begin{proof}
{\bf Step 1: Weighted Hilbert spaces.} Set $\mathsf W_{kl}:=\B(\cR_l,\cR_k)$ and define
\[
  H_s:=\bigl(\mathsf W_{kl},h_s\bigr),
  \qquad 0\le s\le1.
\]
Since $\sigma_k$ and $\sigma_l$ are faithful, each $h_s$ is an inner product. Relative
to the Hilbert--Schmidt inner product, $h_0$ and $h_1$ are respectively given by the
operator
\begin{equation}
  G_0Z=Z\sigma_l,
  \qquad
  G_1Z=\sigma_kZ.
  \label{eq:endpoint-gram}
\end{equation}
These two operators are positive definite and commute. Note that the inner product $h_s$
is given by the positive definite operator
\[
  G_sZ=\sigma_k^sZ\sigma_l^{1-s}
\]
By simultaneously diagonalizing $G_0$ and $G_1$ and the weighted $L_2$ interpolation
formula, we have the isometric identity of complex interpolation
\cite[Theorem~5.5.3]{BerghLofstrom1976}.
\begin{equation}
  [H_0,H_1]_s=H_s,
  \qquad
  0<s<1,
  \label{eq:hilbert-interpolation}
\end{equation}
{\bf Step 2: Centering.} If $k\ne l$, then $V_{kl}=\mathsf W_{kl}$.  If $k=l$, the map
\begin{equation}
  \Pi: Z\mapsto Z-\Tr(\sigma_k Z)I.
  \label{eq:centering-projection}
\end{equation}
is the orthogonal projection onto $V_{kk}$ in every $H_s$. The inclusion
$V_{kk}\hookrightarrow\mathsf W_{kk}$ and the projection $\Pi$ are contractive at both
endpoints.  The operator interpolation theorem \cite[Theorem~4.1.2]{BerghLofstrom1976}
therefore gives the isometric identity
\begin{equation}
  [\Pi H_0,\Pi H_1]_s=\Pi[H_0,H_1]_s=\Pi H_s,
  \qquad 0<s<1.
  \label{eq:centered-interpolation}
\end{equation}

\noindent{\bf Step 3: Endpoint interpolation.} For $Z\in V_{kl}$ and $a=0,1$,
\[
  \frac{\dd}{\dd t}h_a(S_tZ)
  =2\Re h_a(S_tZ,L S_tZ)
  =-2q_a(S_tZ).
\]
Since $V_{kl}$ is invariant under $S_t$, the endpoint bounds \eqref{eq:endpoint-bounds}
and Gr\"onwall's inequality yield
\begin{equation}
  \|S_t|_{V_{kl}}\|_{H_a\to H_a}
  \le e^{-\kappa_a t},
  \qquad a=0,1.
  \label{eq:endpoint-decay}
\end{equation}
For $0<s<1$, interpolating $S_t|_{V_{kl}}$ between the endpoint spaces and using
\eqref{eq:hilbert-interpolation} and \eqref{eq:centered-interpolation} gives
\begin{align}
  \|S_t|_{V_{kl}}\|_{H_s\to H_s}
  \le
  \|S_t|_{V_{kl}}\|_{H_0\to H_0}^{1-s}
  \|S_t|_{V_{kl}}\|_{H_1\to H_1}^{s}\le e^{-((1-s)\kappa_0+s\kappa_1)t}.
  \label{eq:interpolated-decay}
\end{align}
Set $\kappa_s=(1-s)\kappa_0+s\kappa_1$.  Then
\[
  h_s(S_tZ)\le e^{-2\kappa_st}h_s(Z).
\]
Differentiating at $t=0^+$ gives
\[
  -2q_s(Z)
  =\left.\frac{\dd}{\dd t}h_s(S_tZ)\right|_{t=0^+}
  \le-2\kappa_sh_s(Z),
\]
which is \eqref{eq:interpolated-form}.
\end{proof}

\subsection{BKM decomposition and endpoint limits}
\label{subsec:corner-formulas}

Fix a faithful invariant state $\omega$, and choose an orthonormal eigenbasis
$\{|k,i\rangle\}_i$ for each noiseless block $\eta_k$.  In these bases,
\[
  \omega=\bigoplus_k\eta_k\otimes\sigma_k,
  \qquad
  \eta_k=\sum_i\lambda_{ki}|k,i\rangle\langle k,i|.
\]
For $A=A^*$, set $X=\Gamma_\omega^{-1}(A)$.  Then $X=X^*$ and
$H_\omega(A)=\Tr[X\Gamma_\omega(X)]$.  Trace duality and Hilbert--Schmidt
self-adjointness of $\Gamma_\omega$ give
\begin{align*}
  Q_\omega(A)=-\Tr[\cL_*(\Gamma_\omega(X))X]=-\Tr[\Gamma_\omega(X)\cL(X)]
  =-\Re\Tr[X^*\Gamma_\omega(\cL(X))].
\end{align*}
Consider a matrix-unit coordinate
\[
  Y=E_{ij}^{kl}\otimes Z,
  \qquad Z\in\B(\cR_l,\cR_k),
\]
and write
\[
  u=(k,i),\qquad v=(l,j),\qquad
  r=\frac{\lambda_{ki}}{\lambda_{lj}}.
\]
If $u=v$, assume in addition that $Z=Z^*$. Set
\[
  X=
  \begin{cases}
    Y+Y^*,&u\ne v,\\
    Y,&u=v,
  \end{cases}
  \qquad
  A=\Gamma_\omega(X).
\]

\begin{lemma}
\label{lem:corner-bkm}
For the above choices, if $Z\ne0$, then
\begin{equation}
  \frac{Q_\omega(A)}{H_\omega(A)}
  =R_{kl}(r,Z)
  :=\frac{\displaystyle\int_0^1r^s q_{kl,s}(Z)\dd s}
  {\displaystyle\int_0^1r^s h_{kl,s}(Z)\dd s}.
  \label{eq:corner-quotient}
\end{equation}
\end{lemma}

\begin{proof}
By the choice of eigenbases,
\begin{equation}
  \omega^sY\omega^{1-s}
  =\lambda_{lj}r^s E_{ij}^{kl}\otimes
  \sigma_k^sZ\sigma_l^{1-s}.
  \label{eq:weighted-coordinate}
\end{equation}
Together with $\cL(Y)=E_{ij}^{kl}\otimes L_{kl}(Z)$, this gives
\begin{align*}
  \Tr[Y^*\Gamma_\omega(Y)]&=\lambda_{lj}\int_0^1r^sh_{kl,s}(Z)\dd s,
  &\qquad -\Re\Tr[Y^*\Gamma_\omega(\cL(Y))]&=\lambda_{lj}\int_0^1r^sq_{kl,s}(Z)\dd s.
\end{align*}

If $u\ne v$, then $Y$ and $Y^*$ lie in distinct matrix-unit coordinate spaces, which are
orthogonal with respect to the Hilbert--Schmidt inner product. Both $\Gamma_\omega$ and
$\cL$ preserve these coordinates, so the cross terms vanish.  Since both maps also
preserve adjoints, the contributions from $Y$ and $Y^*$ are equal after taking real
parts.  Hence
\begin{align*}
  H_\omega(A)&=2\Tr[Y^*\Gamma_\omega(Y)],
  &\qquad Q_\omega(A)&=-2\Re\Tr[Y^*\Gamma_\omega(\cL(Y))].
\end{align*}

If $u=v$, then $r=1$ and $X=Y=Y^*$, so
\begin{align*}
  H_\omega(A)&=\Tr[Y^*\Gamma_\omega(Y)],
  &\qquad Q_\omega(A)&=-\Re\Tr[Y^*\Gamma_\omega(\cL(Y))].
\end{align*}
Since $Z\ne0$, we have $X\ne0$ and hence $H_\omega(A)>0$. Dividing the expressions for
$Q_\omega(A)$ and $H_\omega(A)$ in either case proves \eqref{eq:corner-quotient}.
\end{proof}

For $r>0$, set
\begin{equation}
  J(r):=\int_0^1r^s\dd s
  =
  \begin{cases}
    (r-1)/\log r,&r\ne1,\\1,&r=1,
  \end{cases}
  \qquad
  \mu_r(\dd s):=\frac{r^s}{J(r)}\dd s.
  \label{eq:mu-r}
\end{equation}

For $Z\ne0$, \eqref{eq:corner-quotient} can be written as
\[
  R_{kl}(r,Z)
  =\frac{\displaystyle
  \int_0^1\frac{q_{kl,s}(Z)}{h_{kl,s}(Z)}
  h_{kl,s}(Z)\,\mu_r(\dd s)}
  {\displaystyle\int_0^1h_{kl,s}(Z)\,\mu_r(\dd s)}.
\]
Thus $R_{kl}(r,Z)$ is the weighted average of the quotients $q_{kl,s}(Z)/h_{kl,s}(Z)$
with respect to the measure $\mu_r(\dd s)$. Note that As $r\to 0^+$,
$\mu_r\Rightarrow\delta_0$, and as $r\to \infty$, $\mu_r\Rightarrow\delta_1$ converges
weakly. Indeed,  write $r=e^{-x}$, where $x>0$.  For continuous $f\in C([0,1])$, the
substitution $t=xs$ gives
\[
  \int_0^1 f(s)\mu_r(\dd s)
  =\frac{1}{1-e^{-x}}\int_0^x e^{-t}f(t/x)\dd t
  \longrightarrow f(0)
  \qquad(x\to\infty).
\]
The limit follows by dominated convergence.  Replacing $s$ with $1-s$ gives the case
$r\to\infty$.

Together with Lemma~\ref{lem:endpoint-interpolation}, these limits determine the infimum
of $R_{kl}$.

\begin{lemma}
\label{lem:corner-endpoints}
For every fixed $0\ne Z\in\B(\cR_l,\cR_k)$,
\begin{align}
  \lim_{r\to 0^+}R_{kl}(r,Z)
  =\frac{q_{kl,0}(Z)}{h_{kl,0}(Z)} \ ,\
  \lim_{r\to \infty}R_{kl}(r,Z)=\frac{q_{kl,1}(Z)}{h_{kl,1}(Z)}.
  \label{eq:corner-limits}
\end{align}
Hence, for every directed corner,
\begin{equation}
  \inf_{\substack{r>0\\0\ne Z\in V_{kl}}}
  R_{kl}(r,Z)
  =\min\left\{
  \kappa_{kl}^{(0)},
  \kappa_{kl}^{(1)}
  \right\}.
  \label{eq:corner-infimum}
\end{equation}
\end{lemma}

\begin{proof}
Applying the weak convergence $$\displaystyle\lim_{r\to 0^+}\mu_r=\delta_0,
\displaystyle\lim_{r\to \infty }\mu_r=\delta_1$$ limits to $s\mapsto q_{kl,s}(Z)$ and
$s\mapsto h_{kl,s}(Z)$, proves the quotient limits. For the second assertion,
Lemma~\ref{lem:endpoint-interpolation} gives
\[
  q_{kl,s}(Z)
  \ge\min\{\kappa_{kl}^{(0)},\kappa_{kl}^{(1)}\}\,h_{kl,s}(Z)
  \qquad \forall \qquad 0\le s\le1,\ Z\in V_{kl},
\]
and integration proves the lower bound.  Note that the limit \eqref{eq:corner-limits}
implies that every $0\ne Z\in V_{kl}$,
\[
  \inf_{r>0}R_{kl}(r,Z)
  \le\frac{q_{kl,a}(Z)}{h_{kl,a}(Z)},
  \qquad a=0,1.
\]
Taking the infimum over $Z\in V_{kl}\setminus\{0\}$ gives the reverse inequality.
\end{proof}

\subsection{Exact BKM coercivity}
\label{subsec:exact-bkm}

The ratios occurring in the BKM corner formula \eqref{eq:corner-quotient} are determined
by the noiseless blocks $\eta_k$  of the invariant state.  A matrix-unit coordinate
$Y=E_{ij}^{kl}\otimes Z$ carries the ratio
\[
  r=\frac{\lambda_{ki}}{\lambda_{lj}}.
\]
As the infimum in Definition~\ref{def:bkm-constants} ranges over all faithful invariant
states, so the noiseless blocks $\eta_k$ may vary.  For a cross-corner $(k\neq l)$,
varying the relative block weights can realize every $r>0$. For a self-corner $k=l$ with
$\dim\cK_k\ge2$, the same holds for coordinates with $i\ne j$. We therefore define the
set of \emph{accessible corners}
\begin{equation}
  \mathsf{Acc}
  :=\{(k,l):k\ne l\}\cup\{(k,k):\dim\cK_k\ge2\},
  \label{eq:accessible-corners}
\end{equation}
and define endpoint constant
\begin{equation}
  \kappa_{\partial}
  :=\min_{\substack{(k,l)\in\mathsf{Acc}\\a\in\{0,1\}}}
  \kappa_{kl}^{(a)}.
  \label{eq:accessible-endpoint}
\end{equation}
Here the subscript $\partial$ refers to the two modular endpoints $s=0,1$. We interpret
a minimum over an empty index set as $+\infty$.

On a self-corner with $\dim\cK_k=1$, we have $r=1$, and define the \emph{internal
constant}
\begin{equation}
  \gamma_k^{\mathrm{int}}
  :=\inf_{\substack{0\ne Z=Z^*\\\Tr(\sigma_kZ)=0}}
  \frac{\displaystyle\int_0^1q_{kk,s}(Z)\dd s}
  {\displaystyle\int_0^1h_{kk,s}(Z)\dd s},
  \label{eq:internal-constant}
\end{equation}
with value $+\infty$ when $V_{kk}=\{0\}$. This is the BKM coercivity constant for the
primitive generator $L_{kk}$ on $B(\mathcal{R}_k,\mathcal{R}_k)$.

\begin{proposition}
\label{prop:corner-decomposition}
Fix a faithful invariant state $\omega$ and the eigenbases from
Section~\ref{subsec:corner-formulas}. Choose a total order on the composite indices
$\{(k,i):k\in\mathcal I,\ 1\le i\le\dim\cK_k\}$. For $A=A^*$, write
\[
  X=\Gamma_\omega^{-1}(A)
  =\sum_{k,l,i,j}E_{ij}^{kl}\otimes Z_{ij}^{kl}.
\]
Then
\[
  E_*A=0
  \quad\Longleftrightarrow\quad EX=0
  \quad\Longleftrightarrow\quad
  Z_{ij}^{kl}\in V_{kl}\quad\text{for all }k,l,i,j.
\]
For $u=(k,i)\le v=(l,j)$, define
\[
  X_{uv}:=
  \begin{cases}
    E_{ij}^{kl}\otimes Z_{ij}^{kl}
    +E_{ji}^{lk}\otimes(Z_{ij}^{kl})^*,&u<v,\\
    E_{ii}^{kk}\otimes Z_{ii}^{kk},&u=v,
  \end{cases}
  \qquad A_{uv}:=\Gamma_\omega(X_{uv}).
\]
Then $A=\sum_{u\le v}A_{uv}$ is a self-adjoint decomposition, and $E_*A=0$ implies
$E_*A_{uv}=0$ for every $u\le v$. Moreover,
\begin{align*}
  H_\omega(A)&=\sum_{u\le v}H_\omega(A_{uv}),
  &\qquad Q_\omega(A)&=\sum_{u\le v}Q_\omega(A_{uv}).
\end{align*}
For $A\ne0$, consequently,
\[
  \frac{Q_\omega(A)}{H_\omega(A)}
  =
  \sum_{\substack{u\le v\\A_{uv}\ne0}}
  \frac{H_\omega(A_{uv})}{H_\omega(A)}
  \frac{Q_\omega(A_{uv})}{H_\omega(A_{uv})},
\]
where the weights are positive and sum to one.
\end{proposition}
\begin{proof}
We first establish
\begin{equation}
  E_*\Gamma_\omega=\Gamma_\omega E.
  \label{eq:expectation-bkm}
\end{equation}
For $k\ne l$, both sides annihilate $E_{ij}^{kl}\otimes Z$. For $k=l$, equations
\eqref{eq:expectation-blocks}--\eqref{eq:predual-blocks} and
$\Tr(\sigma_k^sZ\sigma_k^{1-s})=\Tr(\sigma_kZ)$ give
\begin{align*}
  E_*\Gamma_\omega(E_{ij}^{kk}\otimes Z)
  &=m(\lambda_{ki},\lambda_{kj})\Tr(\sigma_kZ)\,
  E_{ij}^{kk}\otimes\sigma_k\\
  &=\Gamma_\omega E(E_{ij}^{kk}\otimes Z),
\end{align*}
where $m$ is the logarithmic mean. This proves \eqref{eq:expectation-bkm}. Since
$\Gamma_\omega$ is invertible, $E_*A=0$ is equivalent to $EX=0$. By
\eqref{eq:expectation-blocks} and linear independence of the matrix units, $EX=0$ holds
if and only if $\Tr(\sigma_kZ_{ij}^{kk})=0$ for all $k,i,j$, which is precisely the
stated condition on $V_{kl}$.

Since $\Gamma_\omega^{-1}$ preserves adjoints, $X=X^*$ and
$Z_{ji}^{lk}=(Z_{ij}^{kl})^*$. Thus the $X_{uv}$ are self-adjoint and sum to $X$.
Applying $\Gamma_\omega$ gives $A=\sum_{u\le v}A_{uv}$ with $A_{uv}=A_{uv}^*$. If
$EX=0$, the conditions $\Tr(\sigma_kZ_{ij}^{kk})=0$ also give $EX_{uv}=0$ for every
$u\le v$. Equation~\eqref{eq:expectation-bkm} then yields $E_*A_{uv}=0$.

Both $\Gamma_\omega$ and $\cL$ preserve each subspace
$E_{ij}^{kl}\otimes\B(\cR_l,\cR_k)$. For $u<v$, the matrix units in $X_{uv}$ correspond
to $(u,v)$ and $(v,u)$, whereas $X_{uu}$ contains only the matrix unit corresponding to
$(u,u)$. Different components therefore involve different matrix units. Their
Hilbert--Schmidt orthogonality implies that all cross terms between distinct components
vanish. Using $\Gamma_\omega^{-1}(A_{uv})=X_{uv}$, we obtain
\begin{align*}
  H_\omega(A)
  &=\sum_{\substack{u\le v\\u'\le v'}}
  \Tr[X_{uv}^*\Gamma_\omega(X_{u'v'})]
  =\sum_{u\le v}\Tr[X_{uv}^*\Gamma_\omega(X_{uv})]
  =\sum_{u\le v}H_\omega(A_{uv}),\\
  Q_\omega(A)
  &=-\Re\sum_{\substack{u\le v\\u'\le v'}}
  \Tr[X_{uv}^*\Gamma_\omega(\cL(X_{u'v'}))]
  =-\sum_{u\le v}\Re\Tr[X_{uv}^*\Gamma_\omega(\cL(X_{uv}))]
  =\sum_{u\le v}Q_\omega(A_{uv}).
\end{align*}
For $A\ne0$, dividing by $H_\omega(A)$ gives the quotient formula.  Positive
definiteness of $H_\omega$ makes every nonzero-component weight positive, and additivity
of $H_\omega$ shows that the weights sum to one.
\end{proof}
With above proposition, we show that the BKM coercivity constant is characterized by the
endpoint constant and the internal constant together .

\begin{theorem}
\label{thm:exact-bkm}
Let $P_t$ be a quantum Markov semigroup. Assume the limit condition expectation
$\displaystyle E=\lim_{t\to \infty}P_t$ exists and admits the block decomposition as
above. One has
\begin{equation}
  \lambda_{\mathrm{BKM}}
  =\min\left\{
  \kappa_{\partial},
  \min_{k:\,\dim\cK_k=1}
  \gamma_k^{\mathrm{int}}
  \right\}.
  \label{eq:exact-bkm}
\end{equation}
\end{theorem}

\begin{proof}
Set
\[
  m:=\min\left\{
  \kappa_\partial,
  \min_{k:\,\dim\cK_k=1}\gamma_k^{\mathrm{int}}
  \right\}.
\]
Fix a faithful invariant state $\omega$ and a nonzero admissible tangent vector $A$.  By
Proposition~\ref{prop:corner-decomposition}, it suffices to show that every nonzero
component has quotient at least $m$.  On an accessible corner,
\[
  R_{kl}(r,Z)
  \ge\min\{\kappa_{kl}^{(0)},\kappa_{kl}^{(1)}\}
  \ge\kappa_\partial\ge m
\]
by Lemma~\ref{lem:corner-endpoints}.  On a self-corner with $\dim\cK_k=1$, $r=1$ and
$Z=Z^*$, so the quotient is at least $\gamma_k^{\mathrm{int}}\ge m$.  Hence
$\lambda_{\mathrm{BKM}}\ge m$.

For the reverse bound, fix $k$ with $\dim\cK_k=1$ and $V_{kk}\ne\{0\}$. Choose $0\ne
Z=Z^*\in V_{kk}$ and set
\[
  X=E_{11}^{kk}\otimes Z,
  \qquad A=\Gamma_\omega(X).
\]
Proposition~\ref{prop:corner-decomposition} gives $E_*A=0$, and
Lemma~\ref{lem:corner-bkm} identifies the quotient $Q_\omega(A)/H_{\omega}(A)$ with the
ratio in \eqref{eq:internal-constant}.  Taking the infimum over $Z$ gives
$\lambda_{\mathrm{BKM}}\le\gamma_k^{\mathrm{int}}$.

Next fix $(k,l)\in\mathsf{Acc}$, $r>0$, and $0\ne Z\in V_{kl}$. Choose distinct
noiseless basis labels $(k,i)$ and $(l,j)$.  Assign weights $r$ and $1$ to the selected
labels and weight $1$ to every remaining label. Dividing all weights by their sum gives
positive diagonal blocks $\eta_u$ and a faithful invariant state
\[
  \omega=\bigoplus_u\eta_u\otimes\sigma_u,
  \qquad \frac{\lambda_{ki}}{\lambda_{lj}}=r.
\]
Let
\[
  Y:=E_{ij}^{kl}\otimes Z,
  \qquad X:=Y+Y^*,
  \qquad A:=\Gamma_\omega(X).
\]
Proposition~\ref{prop:corner-decomposition} gives $E_*A=0$, and
Lemma~\ref{lem:corner-bkm} gives
\[
  \frac{Q_\omega(A)}{H_\omega(A)}=R_{kl}(r,Z).
\]
Taking the infimum over $r$ and $Z$ and applying Lemma~\ref{lem:corner-endpoints}, we
obtain
\[
  \lambda_{\mathrm{BKM}}
  \le\inf_{\substack{r>0\\0\ne Z\in V_{kl}}}R_{kl}(r,Z)
  =\min\{\kappa_{kl}^{(0)},\kappa_{kl}^{(1)}\}.
\]
Minimizing over all accessible corners gives $\lambda_{\mathrm{BKM}}\le\kappa_\partial$,
proving \eqref{eq:exact-bkm}.
\end{proof}
Under KMS symmetry, Corollary~\ref{cor:internal-positivity}, the finiteness of $\mathcal
I$, and the nonnegativity of the endpoint constants give \eqref{eq:kms-mlsi-criterion}.
\begin{corollary}
\label{cor:internal-positivity}
Let $P_t=e^{t\cL}$ be a quantum Markov semigroup that is KMS symmetric with respect to a
faithful invariant state. Then
\begin{equation}
  \gamma_k^{\mathrm{int}}>0 , \qquad \text{ if } \qquad  \dim\cK_k=1 .
  \label{eq:internal-positive}
\end{equation}
As a consequence,
\begin{equation}
  \alpha(\cL)>0
  \quad\Longleftrightarrow\quad
  \lambda_{\mathrm{BKM}}(\cL)>0
  \quad\Longleftrightarrow\quad
  \kappa_\partial>0.
  \label{eq:kms-mlsi-criterion}
\end{equation}
\end{corollary}

\begin{proof}
If $V_{kk}=\{0\}$, then $\gamma_k^{\mathrm{int}}=+\infty$ by convention.  Otherwise,
\eqref{eq:endpoint-nonnegative} and the decay of $h_{kk,s}(Z)$ in
Lemma~\ref{lem:endpoint-interpolation} gives $q_{kk,s}(Z)\ge0$ for every $s\in[0,1]$ and
$Z\in V_{kk}$. By KMS symmetry and Section~\ref{subsec:corner-generators}, $L_{kk}$ is
self-adjoint for $h_{kk,1/2}$ and has kernel $\mathbb C I$.  Since $V_{kk}$ is the
$h_{kk,1/2}$-orthogonal complement of $\mathbb C I$, there is $\lambda_k>0$ such that
\[
  q_{kk,1/2}(Z)\ge\lambda_k h_{kk,1/2}(Z)
  \qquad (Z\in V_{kk}).
\]
For every nonzero $Z\in V_{kk}$, the integrand is nonnegative and strictly positive at
$s=\tfrac12$.  Continuity in $s$ therefore gives
\[
  \int_0^1q_{kk,s}(Z)\dd s>0,
\]
which implies the quotient
\begin{equation}
  \frac{\displaystyle\int_0^1q_{kk,s}(Z)\dd s}
  {\displaystyle\int_0^1h_{kk,s}(Z)\dd s},
\end{equation}
is positive for every $Z$ on the set
\[
  \left\{Z=Z^*: \Tr(\sigma_kZ)=0,\
  \int_0^1h_{kk,s}(Z)\dd s=1\right\}.
\]
Since this set is compact, the integral quotient attains a positive minimum.
\end{proof}

\section{The complete BKM constant and the GNS spectral gap}
\label{sec:complete-bkm}
In this section, we identify the complete endpoint constant $\kappa_{\partial}^{\mathrm
c}$ with the spectral gap of the GNS symmetrization and prove that the BKM coercivity
constant stabilizes at ancillary dimension two. These identities together give the CMLSI
criterion. We then consider the GNS-symmetric case.
\subsection{The GNS symmetrization and its corner formula}
For a faithful invariant state $\omega$, define the GNS-inner product and norm as
\begin{equation}
  \ip{X}{Y}_{2,\omega}:=\Tr(\omega X^*Y),
  \qquad
  \norm{X}_{2,\omega}^2:=\ip{X}{X}_{2,\omega}.
  \label{eq:gns-product}
\end{equation}
We denote by $L_2(\omega)$ the Hilbert space associated with the above inner product.
Note that the conditional expectation $E$ is the orthogonal projection from
$L_2(\omega)$ onto $\cN$ for every faithful invariant state $\omega$
\cite[Section~2.4]{GaoRouze2022}. Indeed, since $E$ preserves $\omega$ and is
$\cN$-bimodular, for $Y\in\cN$,
\[
  \ip{Y}{X}_{2,\omega}
  =\Tr\!\left(\omega E(Y^*X)\right)
  =\Tr\!\left(\omega Y^*E(X)\right)
  =\ip{Y}{E(X)}_{2,\omega}.
\]
Thus $E$ is self-adjoint on $L_2(\omega)$, and $\B(\cH)=\cN\oplus^{\perp_\omega}\Ker E$.

\Needspace{9\baselineskip}
\begin{proposition}
\label{prop:gns-adjoint}
The $L_2(\omega)$-adjoint of $\cL$ is independent of the choice of faithful invariant
state $\omega$.  Denote this adjoint by $\cL^\sharp$ and define the \emph{GNS
symmetrization} by
\begin{equation}
  \cL_{\mathrm{GNS}}
  :=\frac12(\cL+\cL^\sharp).
  \label{eq:gns-symmetrization}
\end{equation}
The operator $\cL_{\mathrm{GNS}}$ is self-adjoint and non-positive on every
$L_2(\omega)$.  The subspace $\Ker E$ is invariant under $\cL$, $\cL^\sharp$, and
$\cL_{\mathrm{GNS}}$.
\end{proposition}

\begin{proof}
To identify the adjoint of $\cL$, write
\[
  \omega=\bigoplus_k\eta_k\otimes\sigma_k
\]
in the fixed-point decomposition of Proposition~\ref{prop:expectation-structure}. The
directed corners $\mathsf C_{kl}=\B(\cK_l,\cK_k)\otimes\B(\cR_l,\cR_k)$ form an
orthogonal direct sum in $L_2(\omega)$, and \eqref{eq:corner-action} gives
$\cL|_{\mathsf C_{kl}}=\id\otimes L_{kl}$. For $A,B\in\B(\cK_l,\cK_k)$ and
$Z,W\in\B(\cR_l,\cR_k)$,
\begin{align}
  \ip{A\otimes Z}{B\otimes W}_{2,\omega}
  =\Tr(\eta_l A^*B)\,h_{kl,0}(Z,W)\ ,\
  h_{kl,0}(Z,W)=\Tr(Z^*W\sigma_l).
  \label{eq:gns-factorization}
\end{align}
Hence the adjoint of $\cL$ on this corner is $\id\otimes L_{kl}^{\sharp}$, where
$L_{kl}^{\sharp}$ is the adjoint of $L_{kl}$ for $h_{kl,0}$.  Each corner adjoint
depends only on $L_{kl}$ and $h_{kl,0}$, so the adjoint of $\cL$ is independent of the
noiseless blocks $\eta_k$.

Finally, Kadison's inequality and invariance of $\omega$ imply
$\norm{P_tX}_{2,\omega}\le\norm{X}_{2,\omega}$.  Differentiating at zero gives
\[
  \ip{X}{-\cL_{\mathrm{GNS}}X}_{2,\omega}= \frac{1}{2}(\ip{X}{-\cL
  X}_{2,\omega}+\ip{-\cL X}{ X}_{2,\omega})
  =-\Re\ip{X}{\cL X}_{2,\omega}\ge0.
\]
Since $E\cL=\cL E=0$ and $E$ is self-adjoint on $L_2(\omega)$, taking
$L_2(\omega)$-adjoints also gives $E\cL^\sharp=\cL^\sharp E=0$.  Hence $\Ker E$ is
invariant under $\cL$, $\cL^\sharp$, and $\cL_{\mathrm{GNS}}$.
\end{proof}

We note that the operator $\cL_{\mathrm{GNS}}$ need not preserve adjoints and hence need
not generate a QMS; see Section~\ref{subsec:m2-counterexample}.

For a general QMS, extend the definition of the GNS spectral gap by
\begin{equation}
  \lambda_{\mathrm{GNS}}(\cL)
  :=\min\spec\!\left(
  -\cL_{\mathrm{GNS}}\big|_{\Ker E}
  \right)
  =\inf_{\substack{X\ne0\\E(X)=0}}
  \frac{-\Re\Tr(\omega X^*\cL(X))}
  {\Tr(\omega X^*X)}.
  \label{eq:gns-gap}
\end{equation}
The infimum in \eqref{eq:gns-gap} is independent of the chosen faithful invariant state
$\omega$.  If $\cL$ is GNS symmetric, then $\cL^\sharp=\cL$, and
\eqref{eq:gns-gap} recovers the definition in Section~\ref{subsec:qms}. In general,
$\lambda_{\mathrm{GNS}}(\cL)$ can vanish even in finite dimensions
(Section~\ref{subsec:m3-counterexample}).

\begin{remark}
{\rm The above definition agrees with the GNS spectral gap considered in
\cite[Remark~3.9]{WirthKMSGNS2026}. In particular, $\lambda_{\mathrm{GNS}}(\cL)$ is the
largest constant $\lambda$ such that
\[
  \norm{P_tX}_{2,\omega}\le e^{-\lambda t}\norm{X}_{2,\omega},
  \qquad t\ge0,\quad X\in\Ker E.
\]
}
\end{remark}

Define the \emph{complete endpoint constant} by
\begin{equation}
  \kappa_{\partial}^{\mathrm c}:= \inf_{n\ge 1} \kappa_{\partial}(\id_{M_n}\otimes \cL)
  \label{eq:complete-endpoint}
\end{equation}
Here for the matrix amplification $\cL\otimes \id_{M_n}$, $\cK_k$ is replaced by
$\mathbb C^n\otimes\cK_k$, while $\sigma_k$, $L_{kl}$, and $V_{kl}$ remain unchanged.
For $n\ge2$, every nonzero centered corner has two distinct noiseless coordinates, so
its weight ratio $r$ can vary over $(0,\infty)$, which leads the following amplification
stabilization.

\begin{proposition}
\label{prop:gns-corners}
The GNS spectral gap satisfies
\begin{equation}
  \lambda_{\mathrm{GNS}}(\cL)
  =\min_{\substack{k,l\\V_{kl}\ne\{0\}}}
  \kappa_{kl}^{(0)}
  =\kappa_{\partial}^{\mathrm c}.
  \label{eq:gns-corner-gap}
\end{equation}
Consequently, for every $n\ge1$,
\begin{equation}
  \lambda_{\mathrm{GNS}}(\id_{M_n}\otimes\cL)
  =\lambda_{\mathrm{GNS}}(\cL).
  \label{eq:gns-amplification}
\end{equation}
\end{proposition}

\begin{proof}
Choose an eigenbasis of each noiseless block
$\eta_k=\sum_i\lambda_{ki}|k,i\rangle\langle k,i|$ and write
\[
  X=\sum_{k,l,i,j}E_{ij}^{kl}\otimes Z_{ij}^{kl}.
\]
By \eqref{eq:expectation-blocks}, the condition $E(X)=0$ is equivalent to
$Z_{ij}^{kl}\in V_{kl}$ for every component. Orthogonality of the matrix-unit components
and \eqref{eq:corner-action} give
\begin{align}
  \Tr(\omega X^*X)
  &=\sum_{k,l,i,j}\lambda_{lj}
  h_{kl,0}(Z_{ij}^{kl}),\\
  -\Re\Tr(\omega X^*\cL(X))
  &=\sum_{k,l,i,j}\lambda_{lj}
  q_{kl,0}(Z_{ij}^{kl}).
  \label{eq:gns-corner-forms}
\end{align}
Thus the global quotient is a weighted average of the corner quotients
$q_{kl,0}/h_{kl,0}$, with strictly positive weights $\lambda_{lj}h_{kl,0}(Z_{ij}^{kl})$
for the nonzero components.  Conversely, taking $X=E_{ij}^{kl}\otimes Z$ with $0\ne Z\in
V_{kl}$ gives the quotient $q_{kl,0}(Z)/h_{kl,0}(Z)$.  This proves the first equality in
\eqref{eq:gns-corner-gap}.

Since $\cL$ preserves adjoints, $L_{lk}(Z^*)=L_{kl}(Z)^*$, and hence
\begin{equation}
  h_{kl,1}(Z)=h_{lk,0}(Z^*),
  \qquad
  q_{kl,1}(Z)=q_{lk,0}(Z^*).
  \label{eq:reversed-endpoints}
\end{equation}
Since $Z\mapsto Z^*$ is a bijection from $V_{kl}$ onto $V_{lk}$, we have
$\kappa_{kl}^{(1)}=\kappa_{lk}^{(0)}$.  Taking the minimum over $k,l$ and using
\eqref{eq:complete-endpoint} gives the second equality in \eqref{eq:gns-corner-gap}.
Under amplification, $\cK_k$ is replaced by $\mathbb C^n\otimes\cK_k$, while $\sigma_k$,
$L_{kl}$, and $V_{kl}$ remain unchanged.  This proves \eqref{eq:gns-amplification}.
\end{proof}

\subsection{Complete BKM coercivity and the CMLSI criterion}

We now show the complete BKM coercivity constant coincides with GNS spectral gap.

\begin{theorem}
\label{thm:bkm-amplification}
Let $P_t=e^{t\cL}$ be a QMS with a faithful asymptotic conditional expectation $E$. Then
\begin{equation}
  \lambda_{\mathrm{BKM}}^{\mathrm c}(\cL)
  =\lambda_{\mathrm{BKM}}(\id_{M_2}\otimes\cL)
  =\lambda_{\mathrm{GNS}}(\cL).
  \label{eq:complete-bkm-gap}
\end{equation}
\end{theorem}

\begin{proof}
Note that for each amplification $\id_{M_n}\otimes \cL$ with $n\ge 2$, $\cK_k$ is
replaced by $\mathbb C^n\otimes\cK_k$, while $\sigma_k$, $L_{kl}$, and $V_{kl}$ remain
unchanged. Hence the second term in the minimum formula of Theorem \ref{thm:exact-bkm}
always empty. Hence
\[
  \lambda_{\mathrm{BKM}}(\id_{M_n}\otimes\cL)
  =\lambda_{\mathrm{BKM}}(\id_{M_2}\otimes\cL)
  =\min_{\substack{k,l\\V_{kl}\ne\{0\}}}
  \kappa_{kl}^{(0)}=\kappa_\partial^c\ .
\]
which is exactly $\lambda_{\mathrm{GNS}}(\cL)$ by Proposition \ref{prop:gns-corners}.
Testing with product states and product perturbations gives
$\lambda_{\mathrm{BKM}}(\id_{M_2}\otimes\cL)
\le\lambda_{\mathrm{BKM}}(\cL)$.
Taking the infimum over $n\ge1$ completes the proof.
\end{proof}

Substituting the above identity into Corollary~\ref{cor:bkm-mlsi}, we obtain the
following quantitative criterion for complete MLSI constant.

\begin{corollary}[CMLSI criterion]
\label{cor:cmlsi-criterion}
Let $P_t=e^{t\cL}$ be a QMS with a faithful asymptotic conditional expectation $E$. Then
\begin{equation}
  \frac{\lambda_{\mathrm{GNS}}(\cL)}{2C_{\cb}(E)}
  \le\alpha^{\mathrm c}(\cL)
  \le\lambda_{\mathrm{GNS}}(\cL).
  \label{eq:cmlsi-gns-bounds}
\end{equation}
Consequently,
\begin{equation}
  \alpha^{\mathrm c}(\cL)>0
  \quad\Longleftrightarrow\quad
  \lambda_{\mathrm{GNS}}(\cL)>0.
  \label{eq:cmlsi-criterion}
\end{equation}
\end{corollary}
\begin{remark}
The same bound was obtained in \cite{GaoRouze2022} for GNS-symmetric QMS.
Corollary~\ref{cor:cmlsi-criterion} extends this bound to QMS with a faithful asymptotic
conditional expectation.
\end{remark}

\subsection{GNS-symmetric cases}
Let $P_t=e^{\cL t}$ be a quantum Markov semigroup that is GNS symmetric with respect to
a faithful invariant state. It follows that $P_t=e^{\cL t}$ is GNS symmetric with
respect to every faithful invariant state \cite{GaoRouze2022}. Then $\cL$ commutes with
the modular group of every faithful invariant state $\omega$, and also self-adjoint with
respect to the inner product
\begin{equation}
  \ip{X}{Y}_{\omega,s}
  :=\Tr(X^*\omega^sY\omega^{1-s}),
  \label{eq:weighted-products}
\end{equation}
Consequently, all associated coercivity constants, including GNS ($s=0$) and KMS
($s=1/2$),  are determined by the spectrum of $\cL$,
\begin{equation}
  \inf_{\substack{X\ne0\\E(X)=0}}
  \frac{-\ip{X}{\cL(X)}_{\omega,s}}
  {\ip{X}{X}_{\omega,s}}
  =\min(\spec(-\cL)\setminus\{0\})=:\lambda
  \label{eq:weighted-gaps}
\end{equation}
for $s\in [0,1]$. This identifies the GNS and KMS gaps with $\lambda$. Note that $\cL$
is also self-adjoint with respect to the BKM inner product as it is an average of
$s$-inner product. It follows that all spectral gaps coincides with the minimum positive
eigenvalue of $-\cL$.

\begin{corollary}
\label{cor:gns-gaps}
If $\cL$ is GNS symmetric with respect to a faithful invariant state, then
\begin{equation}
  \lambda_{\mathrm{KMS}}(\cL)
  =\lambda_{\mathrm{GNS}}(\cL)
  =\lambda_{\mathrm{BKM}}(\cL)
  =\lambda_{\mathrm{BKM}}^{\mathrm c}(\cL)
  =\min(\spec(-\cL)\setminus\{0\}).
  \label{eq:gns-gaps}
\end{equation}
\end{corollary}
\section{Examples and Counterexamples}
\label{sec:examples}

This section gives two counterexamples separating the spectral and entropic quantities
studied above.  The primitive $M_2$ example has a positive MLSI constant but zero CMLSI
constant, while the non-primitive $M_3$ example has a positive KMS spectral gap but zero
BKM and MLSI constants. Both examples are KMS-symmetric.   We begin with a
Carr\'e-du-Champ identity and a CMLSI criterion for primitive semigroups.

\subsection{A CMLSI criterion for primitive semigroups}
\label{subsec:primitive-criterion}
Let $P_t=e^{t\cL}$ be a QMS with a faithful asymptotic conditional expectation $E$. Let
$\sigma$ be a faithful invariant state.  For $Z\in\B(\cH)$, write
\begin{equation}
  h_0(Z):=\norm{Z}_{2,\sigma}^2=\Tr(\sigma Z^*Z),
  \qquad
  q_0(Z):=-\Re\Tr(\sigma Z^*\cL(Z)).
  \label{eq:primitive-forms}
\end{equation}
Recall the Carr\'e-du-Champ operator defined as
\begin{equation}
  \Gamma_\cL(Z)
  :=\cL(Z^*Z)-\cL(Z^*)Z-Z^*\cL(Z).
  \label{eq:carre-du-champ}
\end{equation}
Suppose that the generator $\cL$ admits the GKLS form of \eqref{eq:gkls}
\begin{equation}
  \cL(X)=i[H,X]
  +\sum_j\left(V_j^*XV_j-\frac12\{V_j^*V_j,X\}\right),
  \label{eq:gkls-jumps}
\end{equation}
where $V_j$ are the jump operators. Then, for every $Z\in\B(\cH)$,
\begin{equation}
  \Gamma_\cL(Z)=\sum_j[V_j,Z]^*[V_j,Z] ,\qquad
  q_0(Z)=\frac12\Tr(\sigma\Gamma_\cL(Z))
  =\frac12\sum_j\norm{[V_j,Z]}_{2,\sigma}^2.
  \label{eq:commutator-form}
\end{equation}
Define the \emph{jump commutant} by
\begin{equation}
  \cD_+(\cL):=\{Z:\Gamma_\cL(Z)=0\}
  =\{Z:[V_j,Z]=0\ \forall j\}.
  \label{eq:jump-commutant}
\end{equation}
For a primitive QMS with faithful invariant state $\sigma$, we have
\begin{align}
  E(X)&=\Tr(\sigma X)I, \qquad
  E_*(\rho)=\Tr(\rho)\sigma,\qquad
  V_\sigma:=\{Z\in\B(\cH):\Tr(\sigma Z)=0\}.
  \label{eq:primitive-data}
\end{align}

\begin{corollary}
\label{cor:primitive-criterion}
Let $\cL$ generate a finite-dimensional primitive QMS that is KMS symmetric with respect
to a faithful invariant state $\sigma$.  Then $\alpha(\cL)>0$.  Moreover,
\begin{equation}
  \alpha^{\mathrm c}(\cL)>0
  \iff \cD_+(\cL)=\mathbb C I.
  \label{eq:primitive-criterion}
\end{equation}
\end{corollary}

\begin{proof}
The positivity of $\alpha(\cL)$ follows from Corollary~\ref{cor:internal-positivity} in
the primitive case. Proposition~\ref{prop:gns-corners} gives
\begin{equation}
  \lambda_{\mathrm{GNS}}(\cL)
  =\kappa_{\partial}^{\mathrm c}(\cL)
  =\inf_{0\ne Z\in V_\sigma}
  \frac{q_0(Z)}{\norm{Z}_{2,\sigma}^2}.
  \label{eq:primitive-gns-gap}
\end{equation}
If $\cD_+(\cL)=\mathbb C I$, the faithfulness of $\sigma$ implies $q_0(Z)>0$ for every
nonzero $Z\in V_\sigma$. Compactness of the unit sphere in $V_\sigma$ then gives the
above infimum $\lambda_{\mathrm{GNS}}(\cL)>0$.  Conversely, if
$Z\in\cD_+(\cL)\setminus\mathbb C I$, then
\[
  Z^\circ:=Z-\Tr(\sigma Z)I\in V_\sigma, \qquad Z^\circ\neq 0, \qquad  q_0(Z^\circ)=0\ .
\]
Thus $\lambda_{\mathrm{GNS}}(\cL)=0$. Corollary~\ref{cor:cmlsi-criterion} then gives the
stated CMLSI criterion.
\end{proof}

\subsection{Positive CMLSI without GNS symmetry on
\texorpdfstring{$M_3$}{M3}}
\label{subsec:m3-positive}

Let
\begin{equation}
  \sigma=\frac16\diag(1,4,1),
  \qquad
  W_1=
  \begin{pmatrix}
    0&1&0\\
    0&0&2\\
    0&0&0
  \end{pmatrix},
  \qquad
  W_2=
  \begin{pmatrix}
    0&0&0\\
    2&0&0\\
    0&1&0
  \end{pmatrix}.
  \label{eq:m3-positive-data}
\end{equation}
Set $K=W_1^*W_1+W_2^*W_2=\diag(4,2,4)$ and define
\begin{equation}
  \cL(X)=W_1^*XW_1+W_2^*XW_2-\frac12\{K,X\}.
  \label{eq:m3-positive-generator}
\end{equation}
Together with $[\sigma^{1/2},K]=0$, the identities
\begin{equation}
  \sigma^{1/2}W_1^*\sigma^{-1/2}=W_2,
  \qquad
  \sigma^{1/2}W_2^*\sigma^{-1/2}=W_1,
  \qquad
  \cL_*(\sigma)=0
  \label{eq:m3-positive-kms}
\end{equation}
show that $\cL$ is KMS symmetric with respect to $\sigma$. For $X=(x_{ij})$, one has
\begin{equation}
  \cL(X)=
  \begin{pmatrix}
    4(x_{22}-x_{11})&2x_{23}-3x_{12}&-4x_{13}\\
    2x_{32}-3x_{21}&x_{11}+x_{33}-2x_{22}&2x_{12}-3x_{23}\\
    -4x_{31}&2x_{21}-3x_{32}&4(x_{22}-x_{33})
  \end{pmatrix}.
  \label{eq:m3-positive-action}
\end{equation}
Hence $\Ker\cL=\mathbb C I$, so the semigroup is primitive, with
\[
  E(X)=\Tr(\sigma X)I,
  \qquad
  E_*(\rho)=\Tr(\rho)\sigma.
\]
It is not GNS symmetric, since
\begin{equation}
  \ip{E_{12}}{\cL(E_{23})}_{\mathrm{GNS},\sigma}=\frac43,
  \qquad
  \ip{\cL(E_{12})}{E_{23}}_{\mathrm{GNS},\sigma}=\frac13.
  \label{eq:m3-positive-not-gns}
\end{equation}

The self-adjoint part of $V_\sigma$ has the real basis
\begin{align}
  Z_0&=\diag(1,0,-1),
  &\cL(Z_0)&=-4Z_0,\nonumber\\
  Z_1&=\diag(2,-1,2),
  &\cL(Z_1)&=-6Z_1,\nonumber\\
  Z_x^\pm&=E_{12}\pm E_{23}+E_{21}\pm E_{32},
  &\cL(Z_x^\pm)&=-(3\mp2)Z_x^\pm,\nonumber\\
  Z_y^\pm&=\mathrm i(E_{12}\pm E_{23}-E_{21}\mp E_{32}),
  &\cL(Z_y^\pm)&=-(3\mp2)Z_y^\pm,\nonumber\\
  Z_{13,x}&=E_{13}+E_{31},
  &\cL(Z_{13,x})&=-4Z_{13,x},\nonumber\\
  Z_{13,y}&=\mathrm i(E_{13}-E_{31}),
  &\cL(Z_{13,y})&=-4Z_{13,y}.
  \label{eq:m3-positive-basis}
\end{align}
These basis vectors are BKM-orthogonal. Indeed, the diagonal directions $Z_0,Z_1$ are
orthogonal, and the adjacent matrix units have the same BKM weight:
\begin{equation}
  \Gamma_\sigma(E_{12})=\frac{E_{12}}{2\log4},
  \qquad
  \Gamma_\sigma(E_{23})=\frac{E_{23}}{2\log4}.
  \label{eq:m3-positive-bkm-weights}
\end{equation}
The same identities hold for their adjoints, while the remaining off-diagonal directions
are orthogonal to them. Since the basis vectors are eigenvectors of $\cL$, the
integrated dissipation form is diagonal as well. Therefore,
\begin{equation}
  \lambda_{\mathrm{BKM}}(\cL)
  =\lambda_{\mathrm{KMS}}(\cL)
  =\min\{4,6,1,5,4\}=1.
  \label{eq:m3-positive-bkm-gap}
\end{equation}

The products
\[
  W_1W_2=2\diag(1,1,0),
  \qquad
  W_2W_1=2\diag(0,1,1)
\]
show that a matrix commuting with both $W_1$ and $W_2$ must be diagonal. Commutation
with $W_1$ then forces its diagonal entries to be equal. Thus $\cD_+(\cL)=\mathbb C I$,
and Corollary~\ref{cor:primitive-criterion} gives $\alpha^{\mathrm c}(\cL)>0$.

To compute the complete BKM constant, consider
\[
  \mathcal W_+=\operatorname{span}_{\mathbb C}\{E_{12},E_{23}\},
  \qquad
  \mathcal W_-=\operatorname{span}_{\mathbb C}\{E_{21},E_{32}\}.
\]
For $Z=xE_{12}+yE_{23}\in\mathcal W_+$, a direct calculation gives
\begin{equation}
  h_0(Z)=\frac{4|x|^2+|y|^2}{6},
  \qquad
  q_0(Z)=\frac{12|x|^2+3|y|^2-10\Re(\overline{x}y)}{6}.
  \label{eq:m3-positive-forms}
\end{equation}
Hence $q_0(Z)\ge\frac12h_0(Z)$, with equality when $y=2x$. On $\mathcal W_-$, the same
formulas hold after interchanging $x$ and $y$, so the minimum quotient is again $1/2$.

The spaces $\mathcal W_+$ and $\mathcal W_-$, the centered diagonal subspace, and
$\operatorname{span}_{\mathbb C}\{E_{13},E_{31}\}$ form a GNS-orthogonal invariant
decomposition of $V_\sigma$. On the latter two subspaces, the quotient $q_0/h_0$ is at
least $4$. The global minimum is therefore $1/2$, attained at $W_1=E_{12}+2E_{23}$.
Equation~\eqref{eq:primitive-gns-gap} gives
\begin{equation}
  \lambda_{\mathrm{GNS}}(\cL)
  =\inf_{0\ne Z\in V_\sigma}\frac{q_0(Z)}{h_0(Z)}
  =\frac{q_0(W_1)}{h_0(W_1)}
  =\frac12.
  \label{eq:m3-positive-gns-gap}
\end{equation}
Theorem~\ref{thm:bkm-amplification} then yields
\begin{equation}
  \lambda_{\mathrm{BKM}}^{\mathrm c}(\cL)
  =\lambda_{\mathrm{GNS}}(\cL)
  =\frac12<\lambda_{\mathrm{BKM}}(\cL)=1.
  \label{eq:m3-positive-complete-gap}
\end{equation}

Finally, Proposition~\ref{prop:pp-indices} gives
\begin{equation}
  C(E)=\|\sigma^{-1}\|_\infty=6,
  \qquad
  C_{\cb}(E)=\Tr(\sigma^{-1})=\frac{27}{2}.
  \label{eq:m3-positive-indices}
\end{equation}
Corollary~\ref{cor:bkm-mlsi} therefore yields
\begin{equation}
  \frac1{12}\le\alpha(\cL)\le1,
  \qquad
  \frac1{54}\le\alpha^{\mathrm c}(\cL)\le\frac12.
  \label{eq:m3-positive-mlsi}
\end{equation}

\subsection{MLSI without CMLSI on \texorpdfstring{$M_2$}{M2}}
\label{subsec:m2-counterexample}

Let
\begin{equation}
  \sigma=\frac15
  \begin{pmatrix}
    1&0\\0&4
  \end{pmatrix}
  ,
  \qquad
  W=
  \begin{pmatrix}
    0& 1/{\sqrt{2}}\\ \sqrt{2}&0
  \end{pmatrix}
  ,
  \qquad
  K=W^*W=\diag\!\left(2,\frac12\right),
  \label{eq:m2-data}
\end{equation}
and define
\begin{equation}
  \cL_2(X)=W^*XW-\frac12\{K,X\}.
  \label{eq:m2-generator}
\end{equation}
The identities
\begin{equation}
  \sigma^{\frac{1}{2}} W^*\sigma^{-\frac{1}{2}}=W,
  \qquad [\sigma^{\frac{1}{2}},K]=0,
  \qquad \cL_{2*}(\sigma)=0
  \label{eq:m2-kms}
\end{equation}
show that $\cL_2$ is KMS symmetric with respect to $\sigma$.  For $X=
\begin{psmallmatrix}
  x&y\\z&t
\end{psmallmatrix}
$,
\begin{equation}
  \cL_2(X)=
  \begin{pmatrix}
    2(t-x)&z-\frac54y\\[
    0.3em]
    y-\frac54z&\frac12(x-t)
  \end{pmatrix}.
  \label{eq:m2-action}
\end{equation}
Hence $\Ker\cL_2=\mathbb C I $, so the semigroup is primitive.  It is not GNS symmetric,
since
\begin{equation}
  \ip{E_{12}}{\cL_2(E_{21})}_{\mathrm{GNS},\sigma}=\frac45,
  \qquad
  \ip{\cL_2(E_{12})}{E_{21}}_{\mathrm{GNS},\sigma}=\frac15.
  \label{eq:m2-not-gns}
\end{equation}
Its GNS symmetrization does not preserve adjoints, since
\[
  \frac12(\cL_2+\cL_2^\sharp)(E_{12}+E_{21})
  =-\frac58E_{12}+\frac54E_{21}.
\]
The self-adjoint part of $V_\sigma$ has the real basis
\begin{align}
  Z_0&=\diag(1,-\frac{1}{4}),&\cL_2(Z_0)&=-\frac52Z_0,
  \label{eq:m2-z0}
  \\
  Z_x&=E_{12}+E_{21},&\cL_2(Z_x)&=-\frac14Z_x,
  \label{eq:m2-zx}
  \\
  Z_y&=\mathrm i(E_{12}-E_{21}),&\cL_2(Z_y)&=-\frac94Z_y.
  \label{eq:m2-zy}
\end{align}
The vector $Z_0$ is BKM-orthogonal to $Z_x$ and $Z_y$, while
\begin{equation}
  h_s(Z_x,Z_y)=\frac{\mathrm i}{5}\bigl(4^{1-s}-4^s\bigr),
  \qquad
  \int_0^1h_s(Z_x,Z_y)\dd s=0.
  \label{eq:m2-bkm-orthogonality}
\end{equation}
The BKM inner product is therefore diagonal in this basis.  Since the basis vectors are
eigenvectors of $\cL_2$, the integrated dissipation form is diagonal as well.  Hence
\begin{equation}
  \lambda_{\mathrm{BKM}}(\cL_2)
  =\lambda_{\mathrm{KMS}}(\cL_2)
  =\min\left\{\frac52,\frac14,\frac94\right\}
  =\frac14.
  \label{eq:m2-bkm-gap}
\end{equation}
Proposition~\ref{prop:pp-indices} gives $C(E)=\|\sigma^{-1}\|_\infty=5$, so
Corollary~\ref{cor:bkm-mlsi} yields
\begin{equation}
  \frac1{40}\le\alpha(\cL_2)\le\frac14.
  \label{eq:m2-mlsi-bounds}
\end{equation}
The operator $W$ satisfies $\Tr(\sigma W)=0$.  Since $[W,W]=0$,
\begin{equation}
  q_0(W)=0,
  \qquad h_0(W)=\Tr(\sigma W^*W)=\frac45>0.
  \label{eq:m2-endpoint-zero}
\end{equation}
We have $W\in\cD_+(\cL_2)$ (non-normality of $W$ gives $W^*\notin\cD_+(\cL_2)$), which
by Corollary \ref{cor:primitive-criterion} implies $\lambda_{\mathrm{GNS}}(\cL_2)=0$.
Theorem~\ref{thm:bkm-amplification} and Proposition~\ref{prop:alpha-upper-bkm} therefore
give
\begin{equation}
  \alpha(\id_{M_2}\otimes\cL_2)
  =\alpha^{\mathrm c}(\cL_2)=0.
  \label{eq:m2-cmlsi-zero}
\end{equation}

We now construct faithful states on $M_2\otimes M_2$ whose MLSI quotients tend to zero.
For $r>0$, set
\begin{equation}
  \eta_r=\frac1{1+r}\diag(r,1),
  \qquad \omega_r=\eta_r\otimes\sigma,
  \label{eq:m2-ancilla-base}
\end{equation}
and
\begin{equation}
  Y=E_{12}\otimes W,
  \qquad X=Y+Y^*,
  \qquad A_r=\Gamma_{\omega_r}(X).
  \label{eq:m2-ancilla-direction}
\end{equation}
Then
\begin{equation}
  A_r=A_r^*\ne0,
  \qquad \Tr(A_r)=0,
  \qquad (\id_{M_2}\otimes E_*)(A_r)=0.
  \label{eq:m2-perturbation-properties}
\end{equation}
A direct calculation gives
\begin{align}
  h_s(W)=\frac25(4^s+4^{-s}), \qquad
  q_s(W)=\frac3{10}(4^s-4^{-s}),
  \label{eq:m2-weighted-forms}
\end{align}
and Lemma~\ref{lem:corner-bkm} yields
\begin{equation}
  R_2(r):=\frac{Q_{\omega_r}(A_r)}{H_{\omega_r}(A_r)}
  =\frac34\,
  \frac{J(4r)-J(r/4)}{J(4r)+J(r/4)}.
  \label{eq:m2-quotient}
\end{equation}
Here $Q_{\omega_r}$ is the dissipation form for $\id_{M_2}\otimes\cL_2$, and $J$ is
defined in \eqref{eq:mu-r}. As $r\downarrow0$,
\begin{equation}
  R_2(r)=\frac{3\log4}{4\log(1/r)}
  +O\!\left(\frac1{\log^2(1/r)}\right)\longrightarrow0.
  \label{eq:m2-quotient-limit}
\end{equation}
For fixed $r>0$, set
\[
  \rho_{r,\delta}:=\omega_r+\delta A_r.
\]
Since $\Tr(A_r)=0$ and $\omega_r>0$, this is a faithful state for every sufficiently
small nonzero real $\delta$.  Moreover, $(\id_{M_2}\otimes
E_*)(\rho_{r,\delta})=\omega_r$. By \eqref{eq:entropy-hessian} and
\eqref{eq:production-hessian},
\[
  \lim_{\delta\to0}
  \frac{\EP_{\id_{M_2}\otimes\cL_2}(\rho_{r,\delta})}
  {2D(\rho_{r,\delta}\Vert\omega_r)}
  =R_2(r).
\]
Choose $r_m\to 0$.  For each $m$, choose $\delta_m\ne0$ sufficiently small that
$\rho_{r_m,\delta_m}$ is faithful and its MLSI quotient differs from $R_2(r_m)$ by less
than $1/m$.  The MLSI quotients of $\rho_{r_m,\delta_m}$ then tend to zero.

\begin{remark}
{\rm In \cite[Theorem~4.10]{GaoJungeLaRacuenteLi2025}, it was proved that GNS-symmetric
quantum Markov semigroups satisfy
\begin{align}
  \alpha^{\mathrm c}\ge \frac{1}{2t_{\cb}(0.1)},
  \label{eq:return-time-bound}
\end{align}
where the complete-positivity return time is defined by
\[
  t_{\cb}(\varepsilon)
  :=\inf\bigl\{t>0:(1-\varepsilon)E\le_{\mathrm{cp}}P_t
  \le_{\mathrm{cp}}(1+\varepsilon)E\bigr\},
  \qquad 0<\varepsilon<1.
\]
Here $\Phi\le_{\mathrm{cp}}\Psi$ means that $\Psi-\Phi$ is completely positive.

The above example shows that \eqref{eq:return-time-bound} does not extend to all
KMS-symmetric semigroups, answering the corresponding question in
\cite[Problem~6.3]{GaoJungeLaRacuenteLi2025} in the negative. Indeed, the Choi matrix of
$E(X)=\Tr(\sigma X)I_2$ is $\sigma^{\mathsf T}\otimes I_2>0$, and $P_t\to E$. For every
$0<\varepsilon<1$, the Choi matrices of $P_t-(1-\varepsilon)E$ and
$(1+\varepsilon)E-P_t$ are therefore positive for all sufficiently large $t$. Hence
$t_{\cb}(\varepsilon)<\infty$. Moreover, $\id\not\le_{\mathrm{cp}}1.1E$, so continuity
at $t=0$ gives $t_{\cb}(0.1)>0$. Thus the right-hand side of
\eqref{eq:return-time-bound} is strictly positive, whereas $\alpha^{\mathrm
c}(\cL_2)=0$. }
\end{remark}

\subsection{A positive KMS gap without MLSI on \texorpdfstring{$M_3$}{M3}}
\label{subsec:m3-counterexample}
Consider
\begin{equation}
  W=
  \begin{pmatrix}
    0&1/\sqrt{2}&0\\
    \sqrt{2}&0&0\\
    0&0&1
  \end{pmatrix},
  \qquad
  K=W^*W=\diag\!\left(2,\frac12,1\right),
  \label{eq:m3-jump}
\end{equation}
and we define
\begin{equation}
  \cL_3(X)=W^*XW-\frac12\{K,X\}.
  \label{eq:m3-generator}
\end{equation}
This is obtained by adjoining the scalar block $1$ to the jump operator in
Section~\ref{subsec:m2-counterexample}. For every $0<\eps<1$, the faithful state
\begin{equation}
  \sigma_\eps
  =\diag\!\left(\frac\eps5,\frac{4\eps}5,1-\eps\right)
  \label{eq:m3-sigma}
\end{equation}
is invariant under $\cL_3$.  The identities $\sigma_\eps^{1/2}W^*=W\sigma_\eps^{1/2}$
and $[\sigma_\eps^{1/2},K]=0$ show that $\cL_3$ is KMS symmetric with respect to
$\sigma_\eps$.

The kernel is
\begin{equation}
  \cN=\Ker\cL_3=\mathbb CP\oplus\mathbb CQ.
  \label{eq:m3-fixed-algebra}
\end{equation}
where $P=E_{11}+E_{22}$ and $Q=E_{33}$. Indeed, if $\cL_3(X)=0$, invariance of
$\sigma_\eps$ and \eqref{eq:commutator-form} imply $[W,X]=0$. Hence
$\cL_3(X)=\frac12[K,X]$, so also $[K,X]=0$.  Since $K$ has distinct diagonal entries,
$X$ is diagonal and $[W,X]=0$ forces $X_{11}=X_{22}$.

The asymptotic expectation and its trace predual are
\begin{align}
  E(X)
  &=\left(\frac15X_{11}+\frac45X_{22}\right)P+X_{33}Q,
  \label{eq:m3-expectation}
  \\
  E_*(\rho)
  &=\Tr(P\rho)\left(\frac15E_{11}+\frac45E_{22}\right)
  +\Tr(Q\rho)Q.
  \label{eq:m3-predual}
\end{align}
The corresponding reference states are $\sigma_P=\diag(1/5,4/5)$ and $\sigma_Q=1$.

The four corners $PM_3P$, $QM_3Q$, $PM_3Q$, and $QM_3P$ are invariant under $\cL_3$.
Its restriction to $PM_3P$ is $\cL_2$, and it vanishes on $QM_3Q$.  The GNS inner
product restricted to $PM_3P$ is $\eps$ times that of
Section~\ref{subsec:m2-counterexample}, so $\cL_3$ is not GNS symmetric. On $PM_3Q$, the
corner matrix is
\begin{equation}
  L_{PQ}
  =
  \begin{pmatrix}
    -\frac32&\sqrt2\\[
    0.3em]
    \frac1{\sqrt2}&-\frac34
  \end{pmatrix},
  \label{eq:m3-corner-map}
\end{equation}
with eigenvalues
\begin{equation}
  \lambda_\pm=\frac{-9\pm\sqrt{73}}8.
  \label{eq:m3-corner-eigenvalues}
\end{equation}
Since $\cL_3$ preserves adjoints, its restriction to $QM_3P$ has the same eigenvalues.
KMS symmetry therefore gives
\begin{equation}
  \lambda_{\mathrm{KMS}}(\cL_3)
  =\min\left\{\frac14,\frac{9-\sqrt{73}}8\right\}
  =\frac{9-\sqrt{73}}8>0.
  \label{eq:m3-kms-gap}
\end{equation}

\Needspace{11\baselineskip}
On the corner $Q\rightarrow P$, take $|u\rangle=2^{-1/2}|1\rangle+|2\rangle\in P\mathbb
C^3$. Since $W|u\rangle=|u\rangle$ and $\langle3|W=\langle3|$, the operator $
|u\rangle\langle3|$ commutes with $W$.  A direct calculation gives
\begin{equation}
  h_{PQ,0}(u)=\frac32,
  \qquad q_{PQ,0}(u)=0,
  \label{eq:m3-endpoint-zero}
\end{equation}
while
\begin{equation}
  h_{PQ,1}(u)=\frac9{10},
  \qquad q_{PQ,1}(u)=\frac3{20}.
  \label{eq:m3-endpoint-one}
\end{equation}
More generally, for every $s\in[0,1]$,
\begin{align*}
  h_{PQ,s}(u)
  =\frac12\left(\frac15\right)^s+\left(\frac45\right)^s, \qquad
  q_{PQ,s}(u)
  =\frac14\left[\left(\frac45\right)^s
  -\left(\frac15\right)^s\right].
\end{align*}
Since $P\ne Q$, this corner is endpoint accessible.  The vanishing of $q_{PQ,0}(u)$,
together with Theorem~\ref{thm:exact-bkm}, gives the  BKM coercivity
$\lambda_{\mathrm{BKM}}(\cL_3)=0$ hence by Corollary \ref{cor:internal-positivity} the
MLSI constant $\alpha (\cL_3)=0$. We also provide direct witness of the $\alpha
(\cL_3)=0$. Set
\[
  X=|u\rangle \langle3|+|3\rangle \langle u| ,
  \qquad A_\eps=\Gamma_{\sigma_\eps}(X).
\]
Then
\[
  A_\eps=A_\eps^*\ne0,
  \qquad E_*A_\eps=0,
  \qquad \Tr(A_\eps)=0.
\]
With $r=\eps/(1-\eps)$, Lemma~\ref{lem:corner-bkm} gives
\begin{equation}
  R_\eps
  :=\frac{Q_{\sigma_\eps}(A_\eps)}{H_{\sigma_\eps}(A_\eps)}
  =\frac{J(4r/5)-J(r/5)}{4J(4r/5)+2J(r/5)}.
  \label{eq:m3-quotient}
\end{equation}
Here $Q_{\sigma_\eps}$ is the dissipation form for $\cL_3$, and $J$ is defined in
\eqref{eq:mu-r}.  As $\eps\to 0$,
\[
  R_\eps
  =\frac{\log2}{3\log((1-\eps)/\eps)}
  +O\!\left(\frac1{\log^2(1/\eps)}\right)
  \longrightarrow0.
\]
For fixed $0<\eps<1$, set
\[
  \rho_{\eps,\delta}:=\sigma_\eps+\delta A_\eps.
\]
This is a faithful state for every sufficiently small nonzero real $\delta$, and
$E_*\rho_{\eps,\delta}=\sigma_\eps$.  By \eqref{eq:entropy-hessian} and
\eqref{eq:production-hessian},
\[
  \lim_{\delta\to0}
  \frac{\EP_{\cL_3}(\rho_{\eps,\delta})}
  {2D(\rho_{\eps,\delta}\Vert\sigma_\eps)}
  =R_\eps.
\]
Choose $\eps_m\to 0$.  For each $m$, choose $\delta_m>0$ sufficiently small that
$\rho_{\eps_m,\delta_m}$ is faithful and its MLSI quotient differs from $R_{\eps_m}$ by
less than $1/m$.  The MLSI quotients of $\rho_{\eps_m,\delta_m}$ tend to zero.
Therefore
\begin{equation}
  \lambda_{\mathrm{BKM}}(\cL_3)=\alpha(\cL_3)=\alpha^{\mathrm c}(\cL_3)=0.
  \label{eq:m3-mlsi-zero}
\end{equation}

\section{Application to the CKG Gibbs sampler}
\label{sec:ckg}

We apply the CMLSI criterion to the Gibbs sampler introduced by Chen, Kastoryano, and
Gily\'en \cite{CKG2025}. This sampler is KMS symmetric with respect to the Gibbs state.
We show that the kernel of its GNS Dirichlet form coincides with its fixed-point
algebra, which implies a positive CMLSI constant by Corollary~\ref{cor:cmlsi-criterion}.

\subsection{The CKG sampler}

Let $H=H^*\in\B(\cH)$ be a Hamiltonian and let $\beta>0$ be inverse temperature. The
Gibbs state is
\begin{equation}
  \sigma_\beta=\frac{e^{-\beta H}}{\Tr(e^{-\beta H})}.
  \label{eq:ckg-gibbs}
\end{equation}
Let $\{A_a:a\in\mathsf A\}$ be a finite family of operators satisfying $A_{a^*}=A_a^*$
for an involution $a\mapsto a^*$ in $\mathsf A$. Write $A_a(t)=e^{itH}A_ae^{-itH}$.

Fix $\sigma_E>0$ and $\omega_\gamma>\beta\sigma_E^2/2$. Following
\cite[Eqs.~(1.3)--(1.4)]{CKG2025}, define the filtered operators
\begin{equation}
  \widehat A_a(\omega)
  =\frac1{\sqrt{2\pi}}\int_{\mathbb R}
  f(t)A_a(t)e^{-i\omega t}\dd t,
  \qquad
  f(t)=\left(\sigma_E\sqrt{\frac2\pi}\right)^{1/2}e^{-\sigma_E^2t^2}.
  \label{eq:ckg-filter}
\end{equation}
The Gaussian transition weight is
\begin{equation}
  \gamma(\omega)=
  \exp\!\left[-\frac{(\omega+\omega_\gamma)^2}{2\sigma_\gamma^2}\right],
  \qquad
  \sigma_\gamma^2=\frac{2\omega_\gamma}{\beta}-\sigma_E^2>0.
  \label{eq:ckg-weight}
\end{equation}
The Lindbladian in Heisenberg picture is
\begin{equation}
  \cL(X)=i[B,X]
  +\sum_a\int_{\mathbb R}\gamma(\omega)
  \left(\widehat A_a(\omega)^*X\widehat A_a(\omega)
  -\frac12\{\widehat A_a(\omega)^*\widehat A_a(\omega),X\}\right)\dd\omega,
  \label{eq:ckg-generator}
\end{equation}
where $B=B^*$ is the Hamiltonian correction from \cite[Lemma~II.1]{CKG2025}, given in
\eqref{eq:ckg-correction}. The semigroup $P_t=e^{t\cL}$, is KMS symmetric with respect
to $\sigma_\beta$ \cite[Corollary~II.2]{CKG2025}. Let $E$ denote its asymptotic
conditional expectation; see Section~\ref{subsec:corner-generators}.

\subsection{The commutants of filtered jump operators}

For $Z\in\B(\cH)$, define the Dirichlet form
\[
  q_0(Z):=-\Re\Tr\!\left(\sigma_\beta Z^*\cL(Z)\right).
\]

\begin{proposition}
\label{prop:ckg-dirichlet-form}
For every $Z\in\B(\cH)$,
\begin{equation}
  q_0(Z)
  =\frac12\sum_a\int_{\mathbb R}\gamma(\omega)
  \norm{[\widehat A_a(\omega),Z]}_{2,\sigma_\beta}^2\dd\omega.
  \label{eq:ckg-dirichlet-form}
\end{equation}
Consequently,
\begin{equation}
  q_0(Z)=0
  \iff [\widehat A_a(\omega),Z]=0
  \quad\text{for every $a$ and a.e. $\omega$.}
  \label{eq:ckg-jump-kernel}
\end{equation}
\end{proposition}

\begin{proof}
The Hamiltonian term cancels in $\Gamma_\cL$, and invariance of $\sigma_\beta$ gives
$q_0(Z)=\frac12\Tr(\sigma_\beta\Gamma_\cL(Z))$. Applying \eqref{eq:commutator-form} with
$V_{a,\omega}=\sqrt{\gamma(\omega)}\,\widehat A_a(\omega)$, with summation over $a$ and
integration over $\omega$, yields \eqref{eq:ckg-dirichlet-form}. The equivalence
\eqref{eq:ckg-jump-kernel} follows from nonnegativity of the integrand, strict
positivity of $\gamma$, and faithfulness of $\sigma_\beta$.
\end{proof}

\Needspace{12\baselineskip}
Let $\cB_H:=\{E'-E:E,E'\in\spec(H)\}$ be the Bohr spectrum of $H$. Write
\begin{equation}
  H=\sum_E EP_E,
  \qquad
  A_{a,\nu}:=\sum_{E'-E=\nu}P_{E'}A_aP_E,
  \qquad
  A_a(t)=\sum_{\nu\in\cB_H}e^{it\nu}A_{a,\nu}.
  \label{eq:ckg-bohr-components}
\end{equation}

\begin{lemma}
\label{lem:ckg-commutants}
The kernel of $q_0$ is given by
\begin{equation}
  \begin{aligned}
    \Ker q_0
    &=\{\widehat A_a(\omega):a\in\mathsf A,\ \omega\in\mathbb R\}'\\
    &=\{A_a(t):a\in\mathsf A,\ t\in\mathbb R\}'\\
    &=\{A_{a,\nu}:a\in\mathsf A,\ \nu\in\cB_H\}'.
  \end{aligned}
  \label{eq:ckg-commutants}
\end{equation}
\end{lemma}

\begin{proof}
Dominated convergence shows that each $\widehat A_a$ is norm continuous. Thus
commutation for a.e. $\omega$ implies commutation for every $\omega$, and
Proposition~\ref{prop:ckg-dirichlet-form} gives the first equality.

Suppose that $Z$ commutes with every $\widehat A_a(\omega)$. For each $a$, the function
$F_a(t):=f(t)[A_a(t),Z]$ belongs to $L_1(\mathbb R;\B(\cH))$ and satisfies
\[
  \widehat F_a(\omega)=[\widehat A_a(\omega),Z]=0.
\]
The Fourier uniqueness theorem, applied entrywise, gives $F_a=0$ a.e.
\cite[Theorem~4.33]{Folland1995}. Since $f\ne0$ a.e., $[A_a(t),Z]=0$ a.e. Norm
continuity then gives this equality for every $t$. The converse follows directly from
\eqref{eq:ckg-filter}, proving the second equality.

For the last equality, the Bohr decomposition \eqref{eq:ckg-bohr-components} shows that
commutation with every $A_{a,\nu}$ implies commutation with every $A_a(t)$. Conversely,
\begin{equation}
  A_{a,\nu}
  =\lim_{T\to\infty}\frac1{2T}\int_{-T}^T
  e^{-it\nu}A_a(t)\dd t,
  \label{eq:ckg-bohr-projection}
\end{equation}
where the limit holds in norm. Hence commutation with every $A_a(t)$ implies commutation
with each $A_{a,\nu}$.
\end{proof}

\Needspace{8\baselineskip}
\subsection{Fixed-point algebra and CMLSI}
Let
\begin{equation}
  \mathfrak A_{\mathrm{orb}}
  :=C^*(A_a(t):a\in\mathsf A,\ t\in\mathbb R).
  \label{eq:ckg-orbit-algebra}
\end{equation}
be the $C^*$-subalgebra generated $A_a(t), t\in R$. Set
\begin{equation}
  R_{\mathrm{diss}}
  :=\sum_a\int_{\mathbb R}\gamma(\omega)
  \widehat A_a(\omega)^*\widehat A_a(\omega)\dd\omega,
  \qquad
  R_{\mathrm{diss},\nu}
  :=\sum_{E'-E=\nu}P_{E'}R_{\mathrm{diss}}P_E.
  \label{eq:ckg-dissipation}
\end{equation}
The Hamiltonian correction from \cite[Lemma~II.1]{CKG2025} is given by function calculus
\begin{equation}
  B=\frac{i}{2}\sum_{\nu\in\cB_H}
  \tanh\!\left(\frac{\beta\nu}{4}\right)R_{\mathrm{diss},\nu}.
  \label{eq:ckg-correction}
\end{equation}

\begin{proposition}
\label{prop:ckg-fixed-algebra}
The fixed-point algebra $\cN:=\Ker\cL=\Fix(P_t)=\Ran E$ satisfies
\begin{equation}
  \cN=\Ker q_0=\mathfrak A_{\mathrm{orb}}'.
  \label{eq:ckg-fixed-algebra}
\end{equation}
The MLSI constants satisfy
\begin{equation}
  \alpha(\cL)\ge\alpha^{\mathrm c}(\cL)
  \ge\frac{\lambda_{\mathrm{GNS}}(\cL)}{2C_{\cb}(E)}>0.
  \label{eq:ckg-cmlsi}
\end{equation}
\end{proposition}

\begin{proof}
Each $\widehat A_a(\omega)$ belongs to $\mathfrak A_{\mathrm{orb}}$, hence so does
$R_{\mathrm{diss}}$. The algebra $\mathfrak A_{\mathrm{orb}}$ is invariant under modular
automorphism group $\alpha_t(\cdot)=e^{itH}\cdot e^{-itH}$, since
$\alpha_t(A_a(s))=A_a(t+s)$. Thus, for every $\nu\in\cB_H$,
\begin{equation}
  R_{\mathrm{diss},\nu}
  =\lim_{T\to\infty}\frac1{2T}\int_{-T}^T
  e^{-it\nu}\alpha_t(R_{\mathrm{diss}})\dd t
  \in\mathfrak A_{\mathrm{orb}},
  \label{eq:ckg-dissipation-bohr}
\end{equation}
where the limit holds in norm. Equation~\eqref{eq:ckg-correction} therefore gives
$B\in\mathfrak A_{\mathrm{orb}}$.

Since $A_a(t)^*=A_{a^*}(t)$, Lemma~\ref{lem:ckg-commutants} gives $\Ker q_0=\mathfrak
A_{\mathrm{orb}}'$. Every $Z$ in this commutant commutes with $B$ and all filtered jump
operators and their adjoints. Hence $\cL(Z)=0$ by \eqref{eq:ckg-generator}. Conversely,
$\cL(Z)=0$ implies $q_0(Z)=0$, proving \eqref{eq:ckg-fixed-algebra}.

Since $\Ker q_0=\cN$ and $\cN\cap\Ker E=\{0\}$, \eqref{eq:gns-gap} and compactness give
\[
  \lambda_{\mathrm{GNS}}(\cL)
  =\min_{\substack{E(Z)=0\\\norm{Z}_{2,\sigma_\beta}=1}}
  q_0(Z)>0.
\]
Corollary~\ref{cor:cmlsi-criterion} and $\alpha^{\mathrm c}(\cL)\le\alpha(\cL)$ now give
\eqref{eq:ckg-cmlsi}.
\end{proof}

\bibliographystyle{amsplain}
\bibliography{references}

\end{document}